\documentclass{article}[12pt]
\usepackage{amsmath,amssymb,amsthm,color,url,framed}
\usepackage{graphicx,enumerate}
\usepackage[margin=2cm]{geometry}

\newtheorem{theorem}{Theorem}
\newtheorem{remark}[theorem]{Remark}

\newtheorem{definition}[theorem]{Definition}
\newtheorem{corollary}[theorem]{Corollary}
\newtheorem{lemma}[theorem]{Lemma}

\def\MM#1{\boldsymbol{#1}}

\numberwithin{equation}{section}

\title{Variational Principles for Stochastic Soliton Dynamics}
\author{Darryl D Holm and Tomasz M Tyranowski
\\Mathematics Department
\\Imperial College London}
\date{26 January 2016\\ $\,$
stochastic soliton dynamics; symmetry reduced variational principles}                                           

\begin{document}
\maketitle
\makeatother

\begin{abstract}
We develop a variational method of deriving stochastic partial differential equations whose solutions follow the flow of a stochastic vector field. As an example in one spatial dimension we numerically simulate singular solutions (peakons) of the stochastically perturbed Camassa-Holm (CH) equation derived using this method. These numerical simulations show that peakon soliton solutions of the stochastically perturbed CH equation persist and provide an interesting laboratory for investigating the sensitivity and accuracy of adding stochasticity to finite dimensional solutions of stochastic partial differential equations (SPDE). In particular, some choices of stochastic perturbations of the peakon dynamics by Wiener noise (canonical Hamiltonian stochastic deformations, or CH-SD) allow peakons to interpenetrate and exchange order on the real line in overtaking collisions, although this behaviour does not occur for other choices of stochastic perturbations which preserve the Euler-Poincar\'e structure of the CH equation (parametric stochastic deformations, or P-SD), and it also does not occur for peakon solutions of the unperturbed deterministic CH equation. The discussion raises issues about the science of stochastic deformations of finite-dimensional approximations of evolutionary PDE and the sensitivity of the resulting solutions to the choices made in stochastic modelling.
\end{abstract}

%
\newpage

\section{Introduction}\label{intro}
Two main approaches have arisen recently for implementing variational principles in stochastic geometric mechanics. In one of them, sometimes called \emph{stochastic deformation} \cite{ArChCr2014}, the Lagrangian in Hamilton's principle is the classical one, but it is evaluated on underlying stochastic processes and their mean derivatives. This perspective was initially motivated by the quantization of classical systems \cite{ChZa2003,Ya1980,Ya1983} and a probabilistic version of Feynman's path integral approach to quantum mechanics. The stochastic deformation approach brings to bear the full theory of stochastic partial differential equations (SPDE). For more details about the history of the applications of this approach, in particular for fluid dynamics, see \cite{ArChCr2014,Ho2015}.

Here we will advocate a simpler, and more restricted approach. Our approach is based on a generalisation in \cite{Ho2015} of earlier work by Bismut \cite{Bi1981}, L\'azaro-Cam\'i and Ortega \cite{LaCa-Or2008}, and Bou-Rabee and Owhadi \cite{BR-O2009} for stochastic ordinary differential equations (SDE), which unifies their Hamiltonian and Lagrangian approaches to temporal stochastic dynamics, and extends them to stochastic partial differential equations (SPDE) in the case of cylindrical noise in which the spatial dependence is \emph{parametric}, while temporal dependence is stochastic. The advantage of this approach is that the parametric spatial dependence of the stochastic dynamical variables in the resulting SPDEs allows essentially finite-dimensional stochastic methods to be applied at each point of space. This feature allows us to safely assume from the onset that all the objects we introduce in this context are semimartingales. To distinguish between the two approaches, we will call the approach taken here \emph{parametric stochastic deformation} (P-SD). 

Parametric stochastic partial differential equations (P-SPDE) result in the present approach by applying P-SD to a deterministic variational principle. These P-SPDE contain a type of multiplicative, cylindrical, Stratonovich noise that depends on both the solution variables and their \emph{spatial gradients}. This unfamiliar feature does not interfere with the passage to the It\^o representation, though, since the space variable is treated merely as a parameter when dealing with cylindrical noise. That is, one may regard the cylindrical noise process as a finite dimensional stochastic process parametrized by $\MM x$ (the space variable).  Then, the Stratonovich equation makes analytical sense pointwise, for each fixed $\MM x$.  Once this is agreed, then the transformation to It\^o by the standard method also makes sense pointwise in space. For more details about P-SPDE explained in a fluid dynamics context, see \cite{Ho2015}.

In this paper we develop an approach for inserting parametric stochastic deformation with cylindrical noise into systems of evolutionary partial differential equations which derive from deterministic variational principles that are invariant under a Lie group action. The corresponding deterministic dynamical systems are called Euler-Poincar\'e equations. The set of Euler-Poincar\'e equations includes the equations of ideal fluid dynamics, which follow from variational principles whose Lagrangians satisfy certain invariance properties under smooth invertible maps (diffeomorphisms) \cite{MaRa1994,HoMaRa1998}. 

\paragraph{Objectives.} This paper has two main objectives. The first objective is the inclusion of parametric stochastic deformation (P-SD) in the variational principle for the EPDiff partial differential equation.%
\footnote{EPDiff is the PDE which arises when the Lagrangian in Hamilton's principle is a functional of continuous Eulerian vector fields, whose flows are smooth invertible maps (diffeomorphisms). EPDiff is the Euler-Poincar\'e equation arising for Lagrangians which are invariant under the diffeomorphism group. When the Lagrangian is chosen to be the $H^1$ norm of the vector fields,  the EPDiff equation becomes the Camassa-Holm equation \cite{CaHo1993}, which is a completely integrable Hamiltonian system and possesses singular soliton solutions known for their peaked shape as \emph{peakons}.} The second objective is the numerical study of the statistical effects of parametric and canonically Hamiltonian stochastic deformations (CH-SD) on the soliton-like solutions of EPDiff which arise in the deterministic case in one spatial dimension, when the Lagrangian in Hamilton's principle is a Sobolev norm on the continuous vector fields. When the $H^1$ norm is chosen this results in the CH equation \cite{CaHo1993}. We study the stochastic generalised Camassa-Holm equation because its singular momentum map \cite{HoMa2004} persists under the stochastic deformations we introduce here and thereby allows its solutions to be investigated on a finite dimensional invariant manifold. 

\subsection{Stochastic EPDiff variational principle}

The action integral for the stochastic variational principle we shall study for EPDiff is a \emph{stochastically constrained} variational principle $\delta S = 0$, with action integral, $S$, given by
\begin{align}
S(u,p,q)  &=
\int \bigg( \ell(u) dt 
+ \left\langle  p\,,\,{dq} + \widetilde{\pounds}_{\widetilde u} q\,\right\rangle_V \bigg)
\,,\label{SVP1}
\end{align}
where $\ell(u)$ is the unperturbed deterministic Lagrangian, written as a functional of velocity vector field $u\in\mathfrak{X}(\mathbb{R}^3)$. The angle brackets  
\begin{equation}
\langle \,p\,,\,q\,\rangle_V:=\int <p(x,t),q(x,t)>dx
\label{L2pairing}
\end{equation}
denote the spatial $L^2$ integral  over the domain of flow of the pairing $<p\,,\,q>$ between stochastic dynamical variables $q$, which take values in a tensor space $V$, and their dual elements $p$ taking values in $V^*$. In \eqref{SVP1}, the quantity $p\in V^*$ is a Lagrange multiplier and $\widetilde{\pounds}_{\widetilde u}q$ is the \emph{stochastic Lie differential} of the dynamical variable $q\in V$ with respect to a stochastic vector field $\widetilde u(x,t)$ which is defined by the following sum of a drift velocity $u(x,t)$ and \emph{Stratonovich} stochastic process with \emph{cylindrical noise} parameterised by spatial position $x$, \cite{Pa2007,Sc1988} 
\begin{equation}
\widetilde u(x,t) = u(x,t)\,dt - \sum_i \xi_i(x)\circ dW_i(t)
\,.\label{vel-vdt}
\end{equation}
We give a precise definition of the stochastic Lie differential $\widetilde{\pounds}_{\widetilde u} q$ in formula \eqref{eq: Stochastic Lie differential definition} in Section~\ref{sec: Stratonovich stochastic EPDiff equation}. One may interpret equation \eqref{vel-vdt} as the decomposition of a vector field $\widetilde u(x,t)$ defined at spatial position $x$ and time $t$ into a time-dependent drift velocity $u(x,t)$ and a stochastic vector field with cylindrical noise (note that for notational convenience we define $\widetilde u$ as a stochastic differential with respect to the time variable; we will use the tilde to denote objects defined as stochastic differentials throughout the rest of this work). The time-independent vector fields $\xi_i(x)$ with $i=1,2,\dots,M$ in the cylindrical stochastic process are usually interpreted as ``diffusivities" of the  stochastic vector field, and the choice of these $M$ quantities must somehow be specified from the physics of the problem to be considered. In the present considerations, a natural choice will arise from the singular momentum map admitted by the deterministic EPDiff equation \cite{HoMa2004}.

The $L^2$ pairing $\left\langle\,\cdot\,,\,\cdot\,\right\rangle_V$ in the Stochastic Variational Principle (SVP) $\delta S=0$ for action functional $S$ in \eqref{SVP1} with  Lagrange multiplier $p\in T^*V$ enforces the \emph{advection condition} that the quantity $q\in V$ is preserved along the Stratonovich stochastic integral curves of the vector field \eqref{vel-vdt} for any tensor space  $V$. Namely, $q$ satisfies the advection condition 
\begin{equation}
{dq} + \widetilde{\pounds}_{\widetilde u} q = 0
\,.\label{advec-cond}
\end{equation}
The advection condition \eqref{advec-cond} for the quantities $q\in V$ may be regarded as a \emph{stochastic constraint} imposed on the variational principle \eqref{SVP1} via the Lagrange multiplier $p$.

%
%
%

At this point, we have introduced parametric Stratonovich stochasticity into the variational principle \eqref{SVP1}  with dynamical variables $(u,q,p)$ through the constraint that the advected quantities $q\in V$ should evolve by following the Stratonovich stochastic vector field $\widetilde u(x,t)$ in equation \eqref{vel-vdt}. This advection law is formulated as a Lie differential with respect to the Stratonovich stochastic vector field acting on a tensor space.
We use the Stratonovich formulation, so the normal rules of calculus apply. 
For mathematical discussions of Lie derivatives with respect to stochastic vector fields, see, e.g., \cite{IkWa1981, KaSh1994}.

\paragraph{Plan of the paper.}
After setting out the general theory of Euler-Poincar\'e evolutionary P-SPDE in $(x,t)\in \mathbb{R}^3\times \mathbb{R}$ in the remainder of this introductory section, we will specialise to one spatial dimension for $(x,t)\in \mathbb{R}\times \mathbb{R}$ and study the corresponding soliton-like solution behaviour for the case when the Lagrangian $\ell(u)$ in the variational principle \eqref{SVP1} is chosen as a Sobolev norm of the vector field $u$.

The objective of the remainder of the paper is to use the Stratonovich stochastic EPDiff Theorem proved below to study the effects of introducing this type of stochasticity on the interactions of the peakon solutions of the CH equation with parametric stochastic deformation (P-SD). The P-SD of the equations of motion for the singular solutions of stochastic EPDiff in one spatial dimension will be introduced via Hamilton's principle in Section \ref{sec: ParametricStochasticDeformations} and their effects on the numerical solutions will be studied thereafter in comparison with a more general canonical Hamiltonian stochastic deformation (CH-SD) in the sense of \cite{Bi1981,LaCa-Or2008}, which includes P-SD but can be more general. 

The Fokker-Planck equations for the probability density evolution associated with P-SD and CH-SD of the EPDiff equation will be discussed in Section \ref{sec: The Fokker-Planck equation}. In Section~\ref{sec: Stochastic variational integrator} we discuss the numerical algorithm we use, and in Section~\ref{sec: Numerical experiments} we present the results of our numerical studies, including sample paths and mean solutions, the probability distribution for crossing of singular solutions on the real line, statistics of the \emph{first crossing time}, effects of noise screening, comparison with other types of noise such as additive noise in the canonical momentum equation and the results of convergence tests for our stochastic variational integrator. Section \ref{sec: Summary and future work} is devoted to a brief summary of results and discussion of some open problems for possible future work.

\subsection{Stratonovich stochastic EPDiff equation}
\label{sec: Stratonovich stochastic EPDiff equation}
If the drift velocity vector field $u(x,t)$ and the diffusivity vector fields $\xi_i(x)$ satisfy some standard measurability and regularity conditions, then the stochastic vector field \eqref{vel-vdt} possesses a pathwise unique stochastic flow $F_{t,s}(x)$. By definition, this flow almost surely satisfies $F_{s,s}(x)=x$ and the integral equation

\begin{equation}
\label{eq: Integral definition of the stochastic flow}
F_{t_2,s}(x)-F_{t_1,s}(x) = \int_{t_1}^{t_2} u\big( F_{t,s}(x), t \big)\,dt - \sum_i \int_{t_1}^{t_2} \xi_i\big( F_{t,s}(x) \big)\circ dW_i(t),
\end{equation}

\noindent
or shortly, in the differential form,

\begin{equation}
\label{eq: Stochastic differential of the stochastic flow}
dF_{t,s}(x) = u\big( F_{t,s}(x), t \big)\,dt - \sum_i \xi_i\big( F_{t,s}(x) \big)\circ dW_i(t).
\end{equation}

\noindent
It can be proved that for fixed $t,s$ this flow is mean-square differentiable with respect to the $x$ argument, and also almost surely is a diffeomorphism (see \cite{ArnoldSDE}, \cite{IkWa1981}, \cite{KloedenPlatenSDE}, \cite{Kunita}). These properties allow us to generalise the differential-geometric notion of the Lie derivative, which we do in the following definition and theorem.

\begin{definition}[Stochastic Lie differential]
\label{thm: Stochastic Lie differential}
Let $q$ be a smooth tensor field. The stochastic Lie differential $\widetilde{\pounds}_{\widetilde u} q$ is the almost surely unique stochastic differential satisfying
\begin{equation}
\label{eq: Stochastic Lie differential definition}
d(F_{t,s}^* q) := F_{t,s}^* \widetilde{\pounds}_{\widetilde u} q.
\end{equation}
\end{definition}

\begin{theorem}
\label{thm: Stochastic Lie differential formula}
The stochastic Lie differential $\widetilde{\pounds}_{\widetilde u} q$ is almost surely unique and given by

\begin{equation}
\label{eq: Stochastic Lie differential formula}
\widetilde{\pounds}_{\widetilde u} q = \pounds_u q \, dt - \sum_i \pounds_{\xi_i}q \circ dW_i(t),
\end{equation}

\noindent
where $\pounds_u$ and $\pounds_{\xi_i}$ are the standard Lie derivatives.

\begin{proof}
The proof is a straightforward generalisation of the standard differential-geometric construction of the Lie derivative of tensor fields (see \cite{AbrahamMarsdenRatiuMTA}). Whenever necessary, we replace time-differentiation with stochastic differentials and use the weak property \eqref{eq: Stochastic differential of the stochastic flow} of the flow. We first prove \eqref{eq: Stochastic Lie differential formula} when $q$ is a smooth, real-valued function. Then $(F_{t,s}^* q)(x) =q\circ F_{t,s}(x)$. Formula \eqref{eq: Stochastic Lie differential formula} is proved by calculating the stochastic differential $d(F_{t,s}^* q)$ using the rules of Stratonovich calculus. Next we consider the case when $q$ is a smooth vector field. Let $G_\lambda$ be the smooth flow of $q$. Then for fixed $t,s$ the flow $H_\lambda$ of the vector field $F_{t,s}^* q$ satisfies $G_\lambda \circ F_{t,s} = F_{t,s}\circ H_\lambda$. From the mean-square differentiability of $F_{t,s}$ we have mean-square differentiability of both sides with respect to $\lambda$. Differentiating both sides with respect to $\lambda$, evaluating at $\lambda=0$, calculating the stochastic differential with respect to $t$ and comparing terms, we obtain formula \eqref{eq: Stochastic Lie differential formula} for vector fields. For the case when $q$ is a differential one-form we use the property $F_{t,s}^*\langle q, v \rangle = \langle F_{t,s}^*q, F_{t,s}^*v \rangle$, where $v$ is an arbitrary smooth vector field and $\langle \cdot,\cdot\rangle$ is the dual pairing between one-forms and vector fields. Calculating the stochastic differential of both sides and using our already established results for functions and vector fields, we prove \eqref{eq: Stochastic Lie differential formula} for differential one-forms. It is now straightforward to complete the proof for a general tensor field. Almost sure uniqueness of $\widetilde{\pounds}_{\widetilde u} q$ follows from our construction and pathwise uniqueness of the flow $F_{t,s}$. 
\end{proof}
\end{theorem}

In order to choose the form of the spatial correlations, or \emph{diffusivities}, $\xi_i(x)$ of the cylindrical Stratonovich stochasticity in \eqref{vel-vdt}, we notice that the action integral for the variational principle in \eqref{SVP1} may be rearranged into the equivalent form
\begin{align}
S(u,p,q)  =
\int \bigg( \ell(u) dt 
+ \left\langle  p\,,\,{dq} + \widetilde{\pounds}_{\widetilde u} q\,\right\rangle_V \bigg)
=
\int \bigg( \ell(u) dt 
+ \left\langle  p\,,\,{dq}\right\rangle_V 
- \left\langle p\diamond q\,,\,{\widetilde u} \,\right\rangle_{\mathfrak{X}} \bigg)
\,,
\label{SVP2}
\end{align}
where we define the diamond operation $(\diamond)$ in the expression $p\diamond q\in \mathfrak{X}^*$ via the real-valued nondegenerate pairing $\langle\,\cdot\,,\,\cdot\,\rangle_{\mathfrak{X}}: \mathfrak{X}^*\times \mathfrak{X}\to\mathbb{R}$ between a vector field $\eta\in \mathfrak{X}$ and its dual under $L^2$ pairing $\mu\in \mathfrak{X}$ as
\begin{align}
\left\langle p\diamond q\,,\,\eta \,\right\rangle_{\mathfrak{X}}
:=
\left\langle  p\,,\,-\pounds_\eta q\right\rangle_V
\,.
\label{diamond-def}
\end{align}
The diamond operation $(\diamond)$ will be instrumental in deriving the Stratonovich form of the stochastic EPDiff equation from the stochastic variational principle for the action integral in \eqref{SVP1}, as stated in the following theorem.

\begin{theorem}[Stratonovich Stochastic EPDiff equation]\label{SEP-thm}\rm$\,$
The parametric Stratonovich stochastic deformation in the action $S(u,p,q)$ for the stochastic variational principle $\delta S = 0$ for EPDiff given by
\begin{align}
S(u,p,q)  &=
\underbrace{\
\int \bigg( \ell(u) - \left\langle p\diamond q\,,\,u(x,t) \,\right\rangle_{\mathfrak{X}} \bigg)dt
}_{\hbox{Lebesgue integral}}
+
\underbrace{\
\int \left\langle p(x,t)\,,\,dq(x,t) \,\right\rangle_V
}_{\hbox{Stratonovich integral wrt $q$}}
+
\underbrace{
\int \sum_{i}  \left\langle p\diamond q\,,\,\xi_i(x) \,\right\rangle_{\mathfrak{X}} \circ dW_i(t)
}_{\hbox{Stratonovich integral}}
\,,\label{SVP2-redux-FIRST}
\end{align}
yields the following Stratonovich form of the \emph{stochastic} EPDiff equation
\begin{align}
dm + \widetilde{\pounds}_{\widetilde u}m 
=
0\,.
\label{SEP-eqns-thm-FIRST}
\end{align}
The momentum density $m(x,t)$ and velocity vector field $u(x,t)$ in \eqref{SEP-eqns-thm-FIRST} are related by $m=\frac{\delta \ell}{\delta u}$ and the stochastic vector field $\widetilde u(x,t)$ is given by
\begin{align}
\widetilde u(x,t) = u(x,t)\,dt - \sum_{i} \xi_i(x) \circ dW_i(t)
\,.
\label{path-thm-FIRST}
\end{align} 
\end{theorem}
\begin{proof}
The first step is to take the elementary variations of the action integral \eqref{SVP2}, to find
\begin{align}
\delta u:\quad
\frac{\delta \ell}{\delta u} - p\diamond q = 0
\,,\quad
\delta p:\quad
dq + \widetilde{\pounds}_{\widetilde u} q   = 0
\,,\quad
\delta q:\quad
- dp 
+ \widetilde{\pounds}_{\widetilde u}^Tp  = 0
\,.
\label{var-eqns-thm-FIRST}
\end{align}
The first equation in \eqref{var-eqns-thm-FIRST} follows from the definition of the diamond operation $(\diamond)$ in \eqref{diamond-def}.
The second and third equations immediately follow from variations of the equivalent form of the action $S(u,p,q)$ in equation \eqref{SVP2} and integrations by parts with vanishing endpoint and boundary conditions.
The governing equation for $m$ will be recovered by using the result of the following Lemma.
\begin{lemma}\rm\label{Lemma-m-eqn}
Together, the three equations in \eqref{var-eqns-thm-FIRST} imply \eqref{SEP-eqns-thm-FIRST}.

\begin{proof}
For an arbitrary $\eta\in \mathfrak{X}$, one computes the pairing
\begin{align}
\begin{split}
\left\langle 
dm  \,,\, \eta 
\right\rangle_{\mathfrak{X}}
&=   
\left\langle 
dp\diamond q + p\diamond dq\,,\, \eta 
\right\rangle_{\mathfrak{X}}
\\
\hbox{By equation \eqref{var-eqns-thm-FIRST} } &=   
\left\langle 
(\widetilde{\pounds}_{\widetilde u}^Tp) \diamond q 
- p\diamond \widetilde{\pounds}_{\widetilde u} q\,,\, \eta 
\right\rangle_{\mathfrak{X}}
\\&=   
\left\langle 
p\,,\, (-\widetilde{\pounds}_{\widetilde u} \pounds_{\eta} + \pounds_{\eta} \widetilde{\pounds}_{\widetilde u} )q\,
\right\rangle_{V}
\\&=   
\left\langle 
p\,,\, {\rm \widetilde{ad}}_{\widetilde u}{\eta}\,q\,
\right\rangle_{V}
=
-\left\langle 
p\diamond q\,,\, {\rm \widetilde{ad}}_{\widetilde u}{\eta}\,
\right\rangle_{\mathfrak{X}}
\\&=
-\left\langle 
 {\rm \widetilde{ad}}^*_{\widetilde u}(p\diamond q)\,,\,{\eta}\,
\right\rangle_{\mathfrak{X}}
=
-\,\Big\langle 
 \widetilde{\pounds}_{\widetilde u}m\,,\,{\eta}\,
\Big\rangle_{\mathfrak{X}}\,,
\end{split}
\label{calc-lem}
\end{align}
where ${\rm \widetilde{ad}}_{\widetilde u} \eta = -[\widetilde{\pounds}_{\widetilde u},\pounds_\eta] = {\rm ad}_{u} \eta \,dt - \sum_i {\rm ad}_{\xi_i} \eta \circ dW_i(t)$ is the stochastic adjoint action differential. Since $\eta\in \mathfrak{X}$ was arbitrary, the last line completes the proof of the Lemma. In the last step we have used the fact that coadjoint action $({\rm ad}^*)$ is identical to Lie-derivative action $(\pounds)$ for vector fields acting on 1-form densities. 
\end{proof}
\end{lemma}
The result of Lemma \ref{Lemma-m-eqn} now produces the $m$-equation in \eqref{SEP-eqns-thm-FIRST} of Theorem \ref{SEP-thm}. This completes the proof.
\end{proof}

\begin{remark}[Multiplicative noise in the 3D vector stochastic EPDiff equation]\rm$\,$\\
In 3D vector notation, the 1-form density $m$ is expressed as $m=\MM{m}\cdot {\rm d}\MM{x}\otimes {\rm d}^3x$ and equation \eqref{SEP-eqns-thm-FIRST} becomes
\begin{align}
d\MM{m} + \Big( \partial_j(\MM{m}u^j)  + m_j\MM{\nabla} u^j  \Big)dt
+ 
\sum_{i}\Big(\partial_j\big(\MM{m} \xi^j_i(x)\big)  
+ m_j\MM{\nabla} \xi^j_i(x) \Big)\circ dW_i(t)
= 0
\,,
\label{SEP-1D}
\end{align}
with $m={\delta \ell}/{\delta u}$.
Importantly, the noise terms in \eqref{SEP-1D} multiply both the solution and its gradient. The latter is not a common form for stochastic PDEs. In addition, both the spatial correlations $\xi_i(x)$ and their derivatives $\nabla\xi_i(x)$ are involved. The effects of these noise terms on the singular solutions of stochastic EPDiff in one spatial dimension will be treated in Section \ref{sec: ParametricStochasticDeformations} and its numerical solutions will be studied thereafter. 
\end{remark}

\subsection{It\^o version of the stochastic EPDiff equation}\label{Strat-Ito-subsec}

In the It\^o version of the stochastic EPDiff equation, noise terms have zero mean, but \textit{additional drift terms arise}. These drift terms are double Lie derivatives, which are diffusive, as shown in \cite{Ho2015} for stochastic fluid dynamics.

The corresponding It\^o forms of the stochastic EPDiff equation in \eqref{SEP-eqns-thm-FIRST} and the second and third equations  in \eqref{var-eqns-thm-FIRST} are found by using It\^o's formula to identify the quadratic covariation terms \cite{Sh1996} as
\begin{align}
\begin{split}
dm + \widehat {\pounds}_{\widetilde u}m 
&=
\frac12
\sum_{j} \pounds_{\xi_j(x)}\left(\pounds_{\xi_j(x)}m\right) \,dt
\,,\\
dq + \widehat {\pounds}_{\widetilde u}q 
&= 
\frac12
\sum_{j} \pounds_{\xi_j(x)}\left(\pounds_{\xi_j(x)}q\right)  \,dt
\,,\\
dp - \widehat {\pounds}_{\widetilde u}^Tp   
&= 
- \frac12
\sum_{j} \pounds_{\xi_j(x)}^T\left(\pounds_{\xi_j(x)}^Tp\right)  \, dt
\,,
\end{split}
\label{ItoEP-eqns}
\end{align}

\noindent
where $\widehat {\pounds}_{\widetilde u}q = \pounds_u q \, dt - \sum_i \pounds_{\xi_i}q \, dW_i(t)$ is an It\^o stochastic differential related to the stochastic Lie differential \eqref{eq: Stochastic Lie differential formula} (it should be noted that $\widehat {\pounds}_{\widetilde u}$ is \emph{not} a Lie differential) and we have used $[dW_i(t),dW_j(t)]=\delta_{ij}dt$ for Brownian motion to identify the quadratic covariation terms as drift terms. For more details about this sort of calculation in the geometric mechanics context, see \cite{Ho2015}. 

\subsection{Legendre transform to Stratonovich stochastic Lie-Poisson Hamilton equations}\label{LPform-sec}

\begin{theorem}[Lie-Poisson representation of Stratonovich stochastic EP equations]\label{LPform-thm}\rm$\,$\\
The Stratonovich stochastic EP system in \eqref{SEP-eqns-thm-FIRST} may  be written equivalently in terms of a standard semidirect product Lie-Poisson Hamiltonian structure \cite{MaRa1994} with a canonical Poisson bracket $\{q,p\}$, as
\begin{align}
\begin{bmatrix}
dm \\ dq \\ dp
\end{bmatrix}
&=
\begin{bmatrix}
-{\rm \widetilde{ad}}^*_{(\,\cdot\,)} m & -(\,\cdot\,)\diamond q & p\diamond(\,\cdot\,)
\\
-\,\widetilde{\pounds}_{(\,\cdot\,)}q & 0 & 1
\\
\widetilde{\pounds}_{(\,\cdot\,)}^Tp & -1 & 0
\end{bmatrix}
\begin{bmatrix}
\delta \widetilde{h} / \delta m \\ \delta \widetilde{h} / \delta q \\ \delta \widetilde{h} / \delta p
\end{bmatrix}
=:
\left\{
\begin{bmatrix}
m \\ q \\ p
\end{bmatrix}
,
\widetilde{h}
\right\}_{LP}
\,,
\label{SLPstructure-thm}
\end{align}
where $\widetilde{h}(m,q,p)$ is a stochastic differential representing the Legendre transform of the stochastic Lagrangian and $\{\,\cdot\,,\,\cdot\,\}_{LP}$ denotes the Lie--Poisson bracket.
\end{theorem}
\begin{proof}
As usual, the Legendre transform of the stochastic Lagrangian determines the stochastic Hamiltonian and its variational derivatives. In a slight abuse of notation, we may write this Legendre transform as
\begin{align}
\begin{split}
\widetilde{h}(m,q,p)  &= \left\langle m\,,\, u  \right\rangle dt - \ell(u,q)dt 
- \sum_{i}  \left\langle m \,,\,\xi_i(x)
\right\rangle_{\mathfrak{X}}\circ dW_i(t)
\\
&=: H(m,q,p)dt 
- \sum_{i}  \left\langle m \,,\,\xi_i(x)
\right\rangle_{\mathfrak{X}}\circ dW_i(t)
\,,\end{split}
\label{Stochastic-Ham}
\end{align}
where now we allow $q$-dependence in the Lagrangian $\ell(u,q)$. Varying the stochastic Hamiltonian in \eqref{Stochastic-Ham} gives 
\begin{align}
\begin{split}
\delta \widetilde{h}(m,q,p)  
&=
\left\langle \delta m\,,\, 
\frac{\delta \widetilde{h}}{\delta m}  \right\rangle 
+
\left\langle  \delta u\,,\, 
m - \frac{\delta \ell}{\delta u}  \right\rangle dt 
+
\left\langle \frac{\delta \widetilde{h}}{\delta q} \,,\, \delta q \right\rangle 
+
\left\langle \frac{\delta \widetilde{h}}{\delta p} \,,\, \delta p \right\rangle 
\\ &= 
\left\langle \delta m\,,\, udt - \sum_i \xi_i(x)\circ dW_i(t)   \right\rangle  
+
\left\langle  \delta u\,,\, 
m - \frac{\delta \ell}{\delta u}  \right\rangle dt 
+ \left\langle -\, \frac{\delta \ell}{\delta q}\,,\, \delta q  \right\rangle dt
+ \left\langle 0 \,,\, \delta p  \right\rangle
\,.\end{split}
.\label{Stochastic-Ham-var-FIRST}
\end{align}
Consequently, the corresponding variational derivatives of the stochastic Hamiltonian are
\begin{align}
\frac{\delta \widetilde{h}}{\delta m} 
= udt - \sum_i \xi_i(x)\circ dW_i(t)
= \widetilde u
\,,\quad
\frac{\delta \widetilde{h}}{\delta q} = - \frac{\delta \ell}{\delta q}dt 
\quad\hbox{and}\quad
\frac{\delta \widetilde{h}}{\delta p} =  0
\,.\label{Stochastic-Ham-var}
\end{align}
The resulting Lie-Poisson Hamiltonian form of the system of Stratonovich stochastic variational  equations in \eqref{SEP-eqns-thm} is then given by
\begin{align}
\begin{bmatrix}
dm \\ dq \\ dp
\end{bmatrix}
&=
\begin{bmatrix}
-{\rm \widetilde{ad}}^*_{(\,\cdot\,)} m & -\,(\,\cdot\,)\diamond q & p\diamond(\,\cdot\,)
\\
-\widetilde{\pounds}_{(\,\cdot\,)}q & 0 & 1
\\
\widetilde{\pounds}_{(\,\cdot\,)}^Tp & -1 & 0
\end{bmatrix}
\begin{bmatrix}
\delta \widetilde{h} / \delta m \\ \delta \widetilde{h} / \delta q \\ \delta \widetilde{h} / \delta p
\end{bmatrix}
\label{SLPstructure}
=
\begin{bmatrix}
-{\rm \widetilde{ad}}^*_{\widetilde u} m - \frac{\delta \widetilde{h}}{\delta q}\diamond q
\\ 
-\widetilde{\pounds}_{\widetilde u}q 
\\ 
\widetilde{\pounds}_{\widetilde u}^Tp - \frac{\delta \widetilde{h}}{\delta q}
\end{bmatrix}.
\end{align}
Of course, the terms involving ${\delta \widetilde{h}}/{\delta q}$ vanish, when ${\delta \ell}/{\delta q}=0$, as for EPDiff.
\end{proof}

\begin{remark}\rm

The matrix operator in \eqref{SLPstructure} is the Hamiltonian operator for a standard semidirect product Lie-Poisson structure \cite{HoMaRa1998} with the canonical Poisson bracket (two-cocycle) $\{q,p\}$ between $q$ and $p$. Similar Lie-Poisson structures also appear in the Hamiltonian formulation of the dynamics of complex fluids \cite{GBRa2009}. 

Stochastic Lie-Poisson Hamiltonian dynamics has also been previously studied and developed in applications in spin dynamics in \cite{BrElHo2008, BrElHo2009} and in image registration in \cite{CoCoVi2012, HoRaTrYo2004, TrVi2012}. 

\emph{Canonical} stochastic Hamilton equations were introduced in Bismut \cite{Bi1981} and were recently developed further in the context of geometric mechanics in L\'azaro-Cam\'i and Ortega \cite{LaCa-Or2008}.
\end{remark}

\section{Stochastic variational perturbations in one spatial dimension}
\label{sec: ParametricStochasticDeformations}

\subsection{Singular peakon solutions of the EPDiff equations}

The EPDiff$(H^1)$ equation in the one-dimensional case when $\ell(u)=\frac12\|u\|^2_{H^1}=\frac12\int u^2 + \alpha^2u_x^2\, {\rm d} x$ is called the Camassa-Holm (CH) equation for $m=\delta \ell/\delta u = u - \alpha^2u_{xx}$ with positive constant $\alpha^2$; namely \cite{CaHo1993},
\begin{align}
m_t + (um)_x + mu_x = 0 
\quad\hbox{with}\quad
m = u - \alpha^2u_{xx}
\,.
\label{CH-eqn}
\end{align}
This equation has singular \emph{peakon} solutions, given by 
\begin{align}
m(x,t):=\frac{\delta \ell}{\delta u} = u - \alpha^2u_{xx} = \sum_{a=1}^N p_a(t) \delta (x - q_a(t)) 
\,,\quad\hbox{so that}\quad
u(x,t): = \sum_{b=1}^N p_b(t) K (x - q_b(t)), 
\label{m-def-peakon}
\end{align}
where $K (x - y)=\exp (-|x-y|/\alpha)$ is the Green's function for the Helmholtz operator $1-\alpha^2\partial_x^2$. The peaked shape of the velocity profile for each individual peakon solution of the CH equation $u(x,t): = p(t) \exp (-|x - q(t)|/\alpha)$ gives them their name. 

Peakons are emergent singular solutions which dominate the initial value problem, as shown in Figure \ref{CH-peakons}. An initially confined smooth velocity distribution will decompose into peakon solutions and, in fact, \emph{only} peakon solutions. Substituting the (weak) solution Ansatz \eqref{m-def-peakon} into the CH equation \eqref{CH-eqn} and integrating against a smooth test function yields the following dynamical equations for the $2N$ solution parameters $q_a(t)$ and $p_a(t)$
\begin{equation}
\frac{dq_a}{dt} = u(q_a(t),t) 
\quad\hbox{and}\quad
\frac{dp_a}{dt} = - \,p_a(t) \frac{\partial u(q_a(t),t)}{\partial q_a} 
\,.
\label{pkn-pq-eqns}
\end{equation}
The system of equations for the peakon parameters comprises a completely integrable canonical Hamiltonian system, whose solutions determine the positions $q_a(t)$ and amplitudes $p_a(t)$, for all $N$ solitons, $a=1,\dots,N$, and also describe the dynamics of their multi body interactions, as shown, for example, in Figure \ref{CH-peakons}. 

\begin{figure}
\begin{center}
\includegraphics[width=.5\textwidth,angle=0]{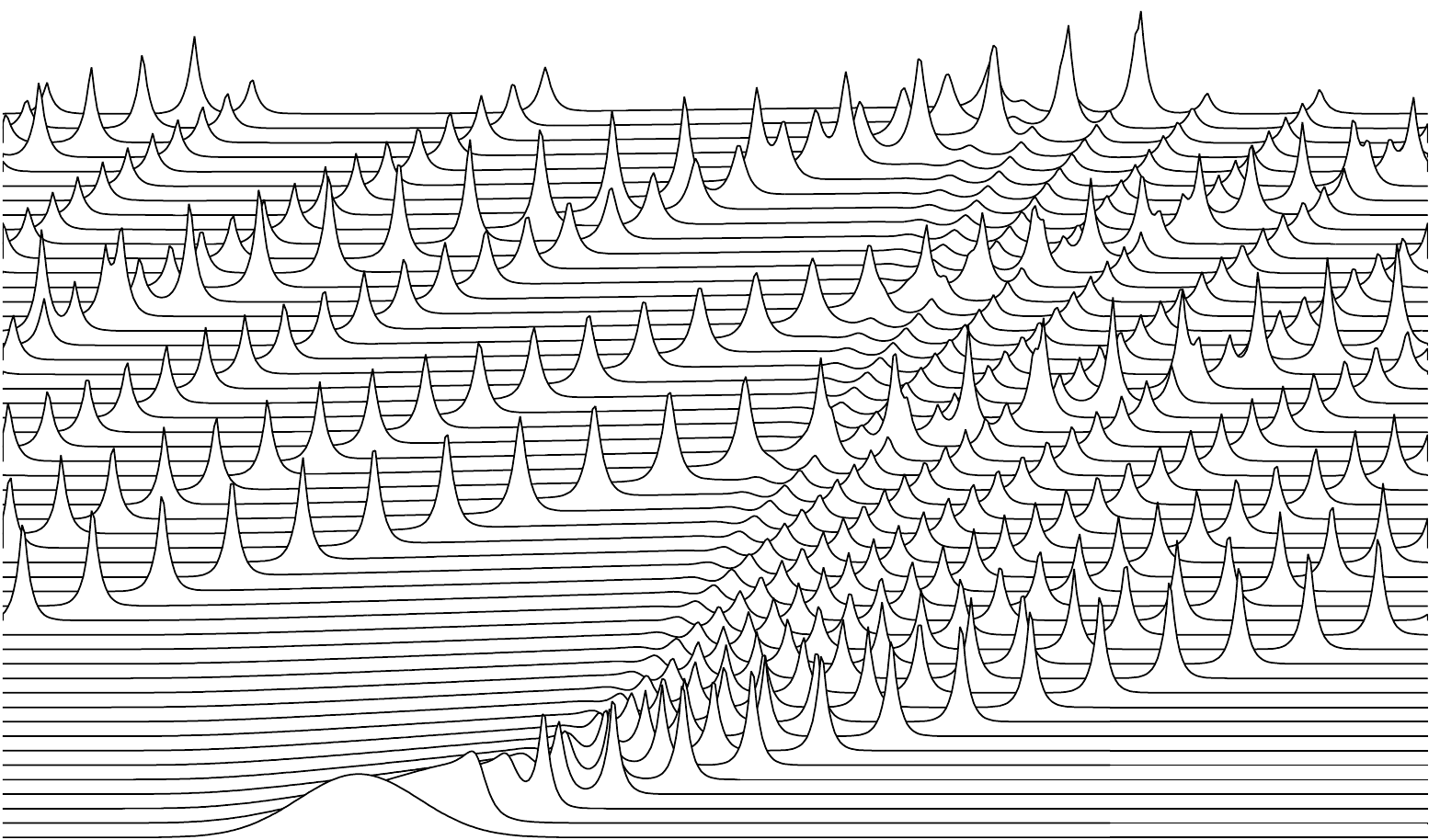}
\end{center}
\caption{\label{CH-peakons} Singular peakon solutions emerge from smooth initial conditions and form a finite dimensional solution set for the CH equation, EPDiff($H^1$). The velocity profile for each individual peakon has the peaked shape given by $u(x,t): = p(t) \exp (-|x - q(t)|/\alpha)$, which is the Green's function for the Helmholtz operator. The main point to notice is that the distance between any two peaks never passes through zero. That is, the peakons keep their order, even after any number of overtaking collisions. (The taller peakons travel faster.)
}
\end{figure}

As mentioned earlier, the objective of the remainder of the paper is to use the Stratonovich stochastic EPDiff Theorem \ref{SEP-thm} to study the effects of introducing this type of stochasticity on the interactions of the peakon solutions of the CH equation with Stratonovich parametric stochastic deformation (P-SD) and canonical Hamiltonian stochastic deformation (CH-SD). The first step is to adapt Theorem \ref{SEP-thm} to accommodate the peakon solutions. For this adaptation, the advection condition \eqref{advec-cond} used previously will be replaced by the definition of peakon velocity as the time derivative of peakon position, as in the first equation in \eqref{pkn-pq-eqns}.

\subsection{Singular momentum map version of the Stratonovich stochastic EPDiff equations}

\begin{theorem}[Canonical Hamiltonian Stochastic Deformation (CH-SD) of EPDiff ]\label{SEP-thm-peakons}\rm$\,$

The action $S(u,p,q)$ for the stochastic variational principle $\delta S = 0$ given by
\begin{align}
S(u,p,q)  &=
\int \bigg( \ell(u)\,dt + \sum_{a} \left\langle  p_a\,,\,dq_a - {u}( q_a, t)\,dt\,\right\rangle \bigg)
-
\int \sum_{i}  h_i(q,p)\circ dW_i(t)
\,,\label{SVP2-redux}
\end{align}
leads to the following Stratonovich form of the \emph{stochastic} EPDiff equation
\begin{align}
\begin{split}
&
dm = - \pounds_u m \, dt + \sum_{i} \big\{m \,,\, h_i(q,p)\big\}  \circ dW_i(t)
\,,\\
&dq_a = u(q_a,t) \,dt + \sum_{i} \big\{q_a \,,\, h_i(q,p)\big\}  \circ dW_i(t)
\,,\\
&dp_a  = - p_a(t) \frac{\partial u}{\partial x}(q_a,t)\,dt 
+ \sum_{i} \big\{p_a \,,\, h_i(q,p)\big\}  \circ dW_i(t) 
\,,
\end{split}
\label{SEP-eqns-thm}
\end{align}
where  the momentum density $m$ and velocity $u$ are given by 
\begin{align}
m(x,t):=\frac{\delta \ell}{\delta u} = \sum_{a=1}^N p_a \delta (x - q_a(t)) 
\,,\quad\hbox{and}\quad
u(x,t): = \sum_{b=1}^N p_b K (x - q_b(t)). 
\label{m-def-thm}
\end{align}
\end{theorem}
\begin{proof}
As in the proof of Theorem \ref{SEP-thm}, the first step is to take the variations of the action integral \eqref{SVP2-redux}, to find
\begin{align}
\begin{split}
\delta u:\quad&
\frac{\delta \ell}{\delta u} - \sum_{a=1}^N p_a \delta (x - q_a(t)) = 0
\,,\\
\delta p:\quad&
dq_a - u(q_a,t) \,dt - \sum_{i} \frac{\partial h_i}{\partial p_a}(q,p) \circ dW_i(t)   = 0
\,,\\
\delta q:\quad&
- dp_a - p_a(t) \frac{\partial u}{\partial x}(q_a,t)\,dt 
- \sum_i \frac{\partial h_i}{\partial q_a}(q,p)\circ dW_i(t)
= 0
\,,\end{split}
\label{var-eqns-thm}
\end{align}
after integrations by parts with vanishing endpoint and boundary conditions.
The first variational equation captures the relation \eqref{m-def-thm}, and latter two equations in \eqref{var-eqns-thm} produce the corresponding equations in \eqref{SEP-eqns-thm}. Substituting the latter two equations in \eqref{var-eqns-thm} into the time derivative of the first one yields the first equation in \eqref{SEP-eqns-thm}.
\end{proof}

The particular choice of the functions $h_i(q,p) = \sum_{a=1}^N p_a \xi_i(q_a)$ reproduces the results of Theorem~\ref{SEP-thm} for parameterised stochastic deformation (P-SD) of the peakon solutions. We summarise this observation in the following Corollary.
\begin{corollary}\rm[P-SD is a special case of CH-SD for EPDiff]
\label{Lemma-m-eqn-peakon}
Given the set of diffusivities $\xi_i(x)$, $i=1,\ldots,M$, let $h_i(q,p) = \sum_{a=1}^N p_a \xi_i(q_a)$. Then the momentum density $m(x,t)$ satisfies the equation
\begin{equation}
\label{m-eqn-peakons}
dm + \widetilde{\pounds}_{\widetilde u}m =0\,, 
\end{equation}
where the stochastic vector field $\widetilde u(x,t)$ is given by the P-SD formula,
\begin{align}
\widetilde u(x,t) = u(x,t) \,dt + \sum_{i} \xi_i(x) \circ dW_i(t)\,.
\label{path-thm}
\end{align} 

\begin{proof}
Specialise to $h_i(q,p) = \sum_{a=1}^N p_a \xi_i(q_a)$ in the first line of equation \eqref{SEP-eqns-thm} in Theorem \ref{SEP-thm-peakons}.
\end{proof}
\end{corollary}

\begin{remark}[Outlook: Comparing results for P-SD and CH-SD]
\label{sec: Remark on P-SD and CH-SD}\rm
In Section~\ref{sec: The Fokker-Planck equation} and Section~\ref{sec: Numerical experiments} we will investigate the effects of choosing between two slightly different types of stochastic potentials on the interaction of two peakons, $N=2$, corresponding to P-SD and CH-SD. The two options are $h^{(1)}_i(q,p) = \sum_{a=1}^N p_a \xi_i(q_a)$ and $h^{(2)}_i(q,p) = \sum_{a=1}^N p_a \varphi_{ia}(q)$, respectively, for $i=1,\ldots,M$. These are both linear in the peakon momenta and in the simplest case they have constant coefficients. We will consider numerical simulations for two cases: $h^{(1)}_1(q,p)= c(p_1+p_2)$ (P-SD for $M=1$) and $h^{(2)}_1(q,p)= \beta_1 p_1$, $h^{(2)}_2(q,p)= \beta_2 p_2$ (CH-SD for $M=2$) with constants $c,\beta_1,\beta_2$. Although these choices for $h^{(1)}_1$ and $h^{(2)}_i$ are very similar, they will produce quite different solution behaviour in our numerical simulations of peakon-peakon overtaking collisions in Section~\ref{sec: Numerical experiments}. 
\end{remark}

\begin{remark}[Stratonovich stochastic EPDiff equations in one dimension]\rm$\,$

\begin{enumerate}
\item
In one spatial dimension, equation \eqref{m-eqn-peakons} becomes
\begin{align}
dm + \big( um_x + 2mu_x\big)dt
+ 
m_x\sum_{i} \xi_i(x) \circ dW_i(t) + 2m\sum_{i} \xi'_i(x) \circ dW_i(t)
= 0
\,.
\label{SEP-1D-1}
\end{align}
Importantly, the multiplicative noise multiplies both the solution and its gradient. The latter is not a common form for stochastic PDEs. In addition, both the spatial correlations $\xi_i(x)$ and their derivatives $\xi'_i(x)$ are involved. 

\item
The equations for $dq_a$ and $dp_a$ in \eqref{SEP-eqns-thm} are stochastic canonical Hamiltonian equations (SCHEs) in the sense of Bismut \cite{Bi1981,LaCa-Or2008}. These equations for $dq_a$ and $dp_a$ may be rewritten as
\begin{align}
\begin{split}
&
dq_a = \frac{\partial H}{\partial p_a}(q,p) \,dt + \sum_{i} \frac{\partial h_i}{\partial p_a}(q,p) \circ dW_i(t)
\,,\\&
dp_a = - \frac{\partial H}{\partial q_a}(q,p)\,dt 
- \sum_i \frac{\partial h_i}{\partial q_a}(q,p) \circ dW_i(t)
\,,
\end{split}
\label{SEP-eqns-thm-qp}
\end{align}
where the deterministic Hamiltonian is given by
\begin{align}
H(q,p) =  \frac12\sum_{a,b} p_ap_bK(q_a-q_b)
\,.
\label{qp-Ham}
\end{align}
The stochastic canonical Hamilton equations \eqref{SEP-eqns-thm-qp} can also be obtained by extremising the \emph{phase-space action functional}
\begin{equation}
\label{eq: Stochastic phase-space action functional}
S\big[q(t),p(t)\big] = \int_0^T \Big( \sum_{a=1}^N p_a \circ dq_a - H(q,p)\,dt \Big) 
- \int_0^T \sum_{i=1}^M h_i(q,p) \circ dW_i(t)\,.
\end{equation}
This is the restriction of \eqref{SVP2-redux} to the submanifold defined by the Ansatz \eqref{m-def-thm}. 
\item
In the It\^o version of stochastic canonical Hamiltonian equations, the noise terms have zero mean, but \textit{additional drift terms arise}. These drift terms are double canonical Poisson brackets, which are diffusive \cite{LaCa-Or2008}:
\begin{align*}
&\delta p :\quad dq_a = u(q_a,t)dt + 
\underbrace{ 
 \sum_i \{q_a,h_i\} dW_i(t) \
 }_{\hbox{It\^o Noise for \it q}}
+ 
\underbrace{ 
\frac12\sum_i \{\{q_a,h_i\},h_i\}dt\
}_{\hbox{It\^o Drift for \it q}},
\\
&\delta q :\quad dp_a = -p_a \frac{\partial u}{\partial x}(q_a,t) \,dt 
+  
\underbrace{ 
\sum_i \{p_a,h_i\} dW_i(t)\
}_{\hbox{It\^o Noise for \it p}} 
 +  
\underbrace{
\frac12\sum_i \{\{p_a,h_i\},h_i\}dt\
}_{\hbox{It\^o Drift for \it p}}.
\end{align*}
The It\^o stochastic dynamics of landmark points in the image registration problem discussed in Trouv\'e and Vialard \cite{TrVi2012} is recovered when we choose $h_i(p,q)=-\sigma q_i$, with $i=1,2,3,$ for $q\in\mathbb{R}^3$ and constant $\sigma$. In that particular case, the double bracket terms vanish. In the present study, we will take $h_i(p,q)$ as the two cases mentioned in Remark~\ref{sec: Remark on P-SD and CH-SD} and compare their effects on the dynamics of \emph{peakons} with $K (x - y)=\exp (-|x-y|/\alpha)$ and \emph{pulsons} with $K (x - y)=\exp (-(x-y)^2/\alpha^2)$ in one spatial dimension. We have introduced the latter Gaussian shaped pulsons, in order to determine how sensitively the numerical results we shall discuss below depend on the jump in derivative of the velocity profile for peakons. 
\end{enumerate}

\end{remark}

\section{The Fokker-Planck equation}
\label{sec: The Fokker-Planck equation}

The stochastic process in \eqref{SEP-eqns-thm-qp} for $(q(t),p(t))$ can be described with the help of a transition density function $\rho(t,q,p; \bar q, \bar p)$ which represents the probability density that the process, initially in the state $(\bar q, \bar p)$, will reach the state $(q,p)$ at time $t$. The transition density function satisfies the Fokker-Planck equation corresponding to \eqref{SEP-eqns-thm-qp} (see \cite{GardinerStochastic}, \cite{KloedenPlatenSDE}). 
Let us examine the form of this equation in the case $h^{(2)}_1(q,p)= \beta_1 p_1$, $h^{(2)}_2(q,p)= \beta_2 p_2$. In that case the noise in \eqref{SEP-eqns-thm-qp} is additive, and the Stratonovich and It\^o calculus yield the same equations of motion.

\subsection{Single-pulson dynamics}
\label{sec: Single-pulson dynamics}

Consider a single pulson ($N=1$) subject to one-dimensional (i.e., $M=1$) Wiener process, with the stochastic potential $h(q,p)=\beta p$, where $\beta$ is a nonnegative real parameter. The stochastic Hamiltonian equations \eqref{SEP-eqns-thm-qp} take the form $dq = p \, dt + \beta \circ dW(t)$, $dp =0,$ which are easily solved by
\begin{equation}
\label{eq: Solution of the stochastic Hamiltonian equations for one pulson}
q_\beta(t) = \bar q + \bar p t + \beta W(t), \qquad \qquad p_\beta(t) = \bar p,
\end{equation}
where $(\bar q, \bar p)$ are the initial conditions. Note that the pulson/peakon retains its initial momentum/height $\bar p$. We will use this solution as a reference for the convergence test in Section~\ref{sec: Convergence test}. The corresponding Fokker-Planck equation takes the form
\begin{equation}
\label{eq: F-P equation for one pulson}
\frac{\partial \rho}{\partial t} + p \frac{\partial \rho}{\partial q} - \frac{1}{2} \beta^2 \frac{\partial^2 \rho}{\partial q^2}=0
\end{equation}
with the initial condition $\rho(0,q,p; \bar q, \bar p) = \delta(q-\bar q)\delta(p-\bar p)$. This advection-diffusion equation is easily solved with the help of the fundamental solution for the heat equation, and the solution yields
\begin{equation}
\label{eq: Solution of F-P equation for one pulson}
\rho_\beta(t,q,p; \bar q, \bar p) = \frac{1}{\beta \sqrt{2 \pi t}} e^{-\frac{(q-\bar q-pt)^2}{2 \beta^2 t}} \delta(p-\bar p).
\end{equation}
This solution means that the initial momentum $\bar p$ is preserved, which is consistent with \eqref{eq: Solution of the stochastic Hamiltonian equations for one pulson}. The position has a Gaussian distribution which widens with time, and whose maximum is advected with velocity $\bar p$.

\subsection{Two-pulson dynamics}
\label{sec: Two-pulson dynamics}

The dynamics of two interacting pulsons has been thoroughly studied and possesses interesting features (see \cite{FringerHolmGeodesic}, \cite{HolmGMS}). 
It is therefore intriguing to see how this dynamics is affected by the presence of noise. Consider $N=2$ pulsons subject to a two-dimensional (i.e., $M=2$) Wiener process, with the stochastic potentials $h_1(q,p) = \beta_1 p_1$ and $h_2(q,p) = \beta_2 p_2$, where $q=(q_1,q_2)$, $p=(p_1,p_2)$, and $\beta_1,\beta_2 \geq 0$. The corresponding Fokker-Planck equation takes the form
\begin{align}
\label{eq: F-P equation for two pulsons}
\frac{\partial \rho}{\partial t} + \frac{\partial}{\partial q_1} \big[a_1(q,p) \rho \big] + \frac{\partial}{\partial q_2} \big[a_2(q,p) \rho \big] + \frac{\partial}{\partial p_1} \big[a_3(q,p) \rho \big] + \frac{\partial}{\partial p_2} \big[a_4(q,p) \rho \big] - \frac{1}{2} \beta_1^2 \frac{\partial^2 \rho}{\partial q_1^2} - \frac{1}{2} \beta_2^2 \frac{\partial^2 \rho}{\partial q_2^2}=0
\end{align}
with the initial condition $\rho\big(0,q,p; \bar q, \bar p \big) = \delta(q_1-\bar q_1) \delta(p_1-\bar p_1) + \delta(q_2-\bar q_2) \delta(p_2-\bar p_2)$, where
\begin{align}
\label{eq: Drift functions}
\begin{split}
a_1(q,p) &= p_1 + p_2 K(q_1-q_2), \qquad\qquad a_3(q,p) = -p_1 p_2 K'(q_1-q_2),  \\
a_2(q,p) &= p_2 + p_1 K(q_1-q_2), \qquad\qquad a_4(q,p) = p_1 p_2 K'(q_1-q_2).
\end{split}
\end{align}
Despite its relatively simple structure, it does not appear to be possible to solve this equation analytically. It is nevertheless an elementary exercise to verify that the function
\begin{equation}
\label{eq: Asymptotic solution to the F-P equation for two pulsons}
\rho(t,q_1,q_2,p_1,p_2; \bar q_1, \bar q_2, \bar p_1, \bar p_2) = \rho_{\beta_1}(t,q_1,p_1; \bar q_1, \bar p_1) + \rho_{\beta_2}(t,q_2,p_2; \bar q_2, \bar p_2),
\end{equation} 
where $\rho_{\beta_i}$ is given by \eqref{eq: Solution of F-P equation for one pulson}, satisfies \eqref{eq: F-P equation for two pulsons} asymptotically as $q_1-q_2 \longrightarrow \pm \infty$, assuming the Green's function and its derivative decay in that limit. This simple observation gives us an intuition that stochastic pulsons should behave like individual particles when they are far from each other, just like in the deterministic case. In order to study the stochastic dynamics of the collision of pulsons, we need to resort to Monte Carlo simulations. 

In Section~\ref{sec: Stochastic variational integrator} we discuss our numerical algorithm, and in Section~\ref{sec: Numerical experiments} we present the results of our numerical studies.

\subsection{Two-pulson dynamics with P-SD}
\label{sec: Two-pulson dynamics with P-SD}

The stochastic interaction of two (or more) pulsons can be analysed explicitly when the stochastic potential has the form $h(q,p)=\beta (p_1+p_2)$ for $\beta\geq 0$ (P-SD for $M=1$; see Remark~\ref{sec: Remark on P-SD and CH-SD}). It is an elementary exercise to show that in that case the stochastic Hamiltonian equations \eqref{SEP-eqns-thm-qp} are solved by

\begin{equation}
\label{eq: P-SD exact solution}
q_a(t) = q^D_a(t) + \beta W(t), \qquad\qquad p_a(t) = p^D_a(t), \qquad\qquad a=1,2,
\end{equation}

\noindent
where $q^D_a(t)$ and $p^D_a(t)$ are the solutions of the deterministic system \eqref{pkn-pq-eqns}. We verify this numerically in Section~\ref{eq: Additive noise in the momentum equation}.

\section{Stochastic variational integrator}
\label{sec: Stochastic variational integrator}

Given the variational structure of the problem we have formulated in Theorem \ref{SEP-thm-peakons}, it is natural to employ variational integrators for numerical simulations. For an extensive review of variational integrators we refer the reader to Marsden \& West \cite{MarsdenWestVarInt} and the references therein. Stochastic variational integrators were first introduced in Bou-Rabee \& Owhadi \cite{BR-O2009}. These integrators were derived for Lagrangian systems using the Hamilton-Pontryagin variational principle. In our case, however, we find it more convenient to stay on the Hamiltonian side and use the discrete variational Hamiltonian mechanics introduced in Lall \& West \cite{LallWestHamiltonian}. We combine the ideas of \cite{BR-O2009} and \cite{LallWestHamiltonian}, and propose the following discretization of the phase-space action functional \eqref{eq: Stochastic phase-space action functional}:


%
\begin{equation}
\label{eq: Discrete action}
S_d = \sum_{k=0}^{K-1} \sum_{i=1}^N  \bigg(p_i^k (q_i^{k+1}-q_i^k) - H(q^{k+1},p^k)\Delta t \bigg) -  \sum_{k=0}^{K-1} \sum_{m=1}^M \frac{h_m(q^k,p^k) + h_m(q^{k+1},p^{k+1})}{2} \Delta W_k^m,
\end{equation}
where $\Delta t = T/K$ is the time step, $(q^k, p^k)$ denote the position and momentum at time $t_k = k \Delta t$, and $\Delta W^m_k \sim N(0,\Delta t)$ are independent normally distributed random variables for $m=1,\ldots,M$ and $k=0,\ldots, K-1$. Let $L(q, \dot q)$ denote the Lagrangian related to $H(q,p)$ via the standard Legendre transform $\dot q_i = \partial H / \partial p_i$. Then one can easily show that \eqref{eq: Discrete action} is equivalent to the discretization
\begin{equation}
\label{eq: Discrete H-P action}
S_d = \sum_{k=0}^{K-1} \bigg( L(q^k,v^k) + \sum_{i=1}^N  p_i^{k+1} \Big( \frac{q_i^{k+1}-q_i^k}{\Delta t} - v^{k+1} \Big) \bigg) \Delta t -  \sum_{k=0}^{K-1} \sum_{m=1}^M \frac{h_m(q^k,p^k) + h_m(q^{k+1},p^{k+1})}{2} \Delta W_k^m
\end{equation}
of the Hamilton-Pontryagin principle used in \cite{BR-O2009}. Omitting the details, \eqref{eq: Discrete action} is obtained by computing the left discrete Hamiltonian $H^-(q^{k+1},p^k)$ corresponding to the discrete Lagrangian $L_d(q^k, q^{k+1}) = \Delta t L(q^{k+1},(q^{k+1}-q^k)/\Delta t)$. The interested reader is referred to \cite{LallWestHamiltonian} for more details. Extremizing \eqref{eq: Discrete action} with respect to $q^k$ and $p^k$ yields the following implicit stochastic variational integrator:
\begin{align}
\label{eq: Stochastic Variational Integrator}
\frac{q_i^{k+1} - q_i^k}{\Delta t} &= \frac{\partial H}{\partial p_i}\big(q^{k+1},p^k\big) +\sum_{m=1}^M \frac{\partial h_m}{\partial p_i}\big(q^k,p^k\big) \frac{\Delta W_{k-1}^m+\Delta W_k^m}{2 \Delta t}, \nonumber \\
\frac{p_i^{k+1} - p_i^k}{\Delta t} &= -\frac{\partial H}{\partial q_i}\big(q^{k+1},p^k\big) -\sum_{m=1}^M \frac{\partial h_m}{\partial q_i}\big(q^{k+1},p^{k+1}\big) \frac{\Delta W_k^m+\Delta W_{k+1}^m}{2 \Delta t},
\end{align}
for $i=1,\ldots, N$. Knowing $(q^k,p^k)$ at time $t_k$, the system above allows to solve for the position $q^{k+1}$ and momentum $p^{k+1}$ at the next time step. For increased computational efficiency, it is advisable to solve the first (nonlinear) equation for $q^{k+1}$ first, and then the second equation for $p^{k+1}$. 

Note that in \eqref{eq: Discrete action} we used $p_i^k (q_i^{k+1}-q_i^k)$ to approximate the Stratonovich integral $\int_{t_k}^{t_{k+1}}p_i\circ dq_i$ in \eqref{eq: Stochastic phase-space action functional}, which means the numerical scheme \eqref{eq: Stochastic Variational Integrator} will not be convergent for general nonlinear stochastic potentials $h_i(q,p)$ (i.e., multiplicative noise). If we used $(p_i^k+p_i^{k+1})(q_i^{k+1}-q_i^k)/2$, the resulting integrator would be a two-step method, i.e., it would not be self-starting, and its geometric/symplectic properties would be in question. Nevertheless, for additive noise, i.e., when the stochastic potentials $h_i(q,p)$ are linear in their arguments, the integrator \eqref{eq: Stochastic Variational Integrator} is a simple modification of the Euler-Maruyama method and its convergence can be established using similar techniques (see \cite{Milstein1995}). The integrator \eqref{eq: Stochastic Variational Integrator} has strong order of convergence 0.5, and weak order of convergence 1. We further verify this fact numerically in Section~\ref{sec: Convergence test}.

The integrator \eqref{eq: Stochastic Variational Integrator} is symplectic, and preserves momentum maps corresponding to (discrete) symmetries of the discrete Hamiltonian---for instance, if $H(q,p)$ and all $h_i(q,p)$ are translationally invariant, as in our simulations in Section~\ref{sec: Numerical experiments}, then the total momentum $\sum_{i=1}^N p_i$ is numerically preserved. The proof of these facts trivially follows from \cite{BR-O2009}, keeping in mind that the momenta $p_i$ and velocities $\dot q_i$ are related via the Legendre transform.

\section{Numerical experiments}
\label{sec: Numerical experiments}

We performed numerical simulations of the rear-end collision of two pulsons for two different Green's functions, namely $K(q_1-q_2) = e^{-(q_1-q_2)^2}$ and $K(q_1-q_2) = e^{-2|q_1-q_2|}$. In the latter case, the corresponding pulsons are commonly called \textquoteleft peakons'. We investigated the initial conditions  $\bar q_1 = 0$, $\bar q_2 = 10$, $\bar p_2 =1$ together with the following four initial values: $\bar p_1 = 8$, $\bar p_1 = 4$, $\bar p_1 = 2$, $\bar p_1 = 1.$ That is, we varied the initial momentum of the faster pulson. We perturbed the slower pulson by introducing a one-dimensional Wiener process with the stochastic potential $h(q,p)=\beta p_2$ (this corresponds to $\beta_1=0$, $\beta_2=\beta$ in Section~\ref{sec: Two-pulson dynamics}). The pulsons were initially well-separated, so their initial evolution was described by \eqref{eq: Asymptotic solution to the F-P equation for two pulsons}. The parameter $\beta$ was varied in the range $[0, 6.5]$. We used the time step $\Delta t = 0.02$, and for each choice of the parameters 50000 sample solutions were computed until the time $T=100$.

\subsection{Sample paths and mean solutions}
\label{sec: Sample paths and mean solutions}

\begin{figure}[h!]
	\centering
		\includegraphics[width=0.8\textwidth]{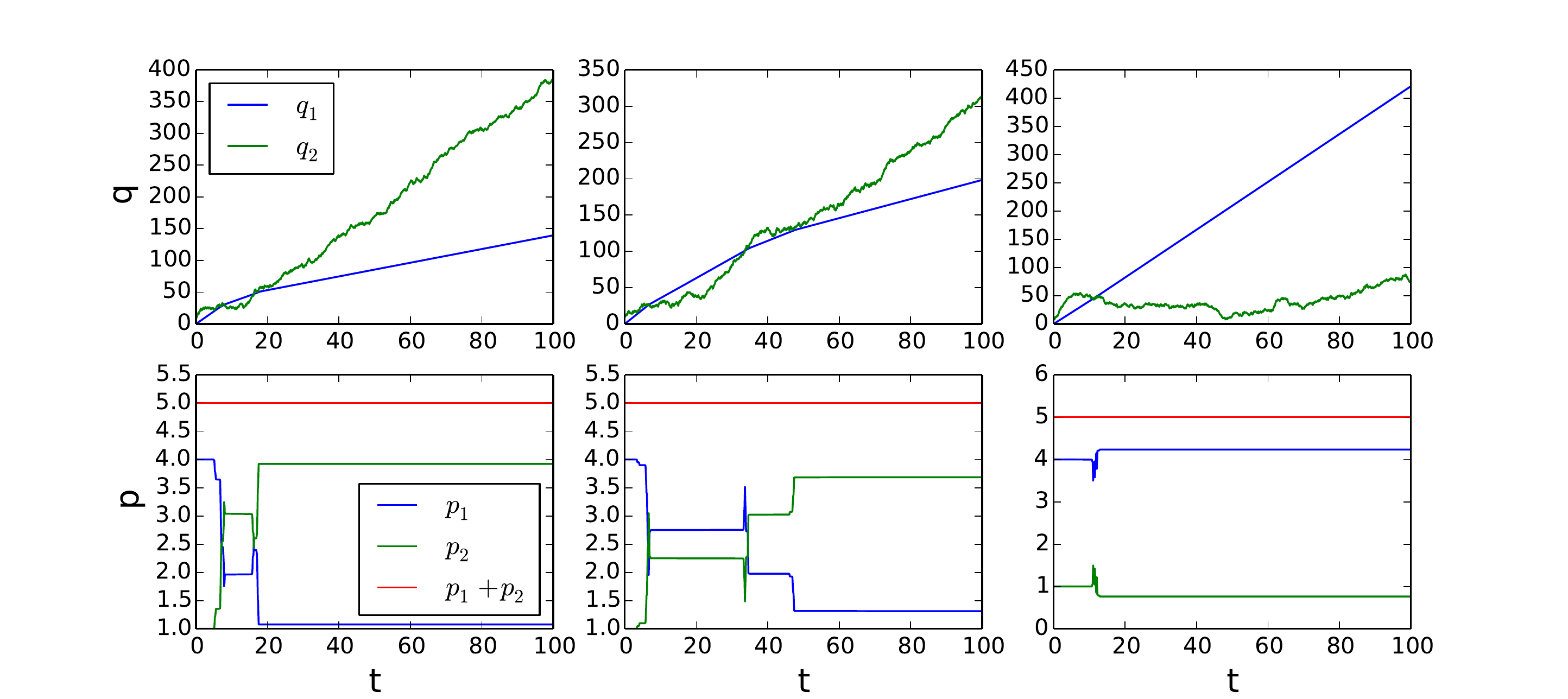}
		\caption{Example numerical sample paths for Gaussian pulsons for the simulations with $\bar p_1=4$ and $\beta=4$. The positions are depicted in the plots in the upper row, and the corresponding momenta are shown in the plots in the lower row.}
		\label{fig: Sample paths 1}
\end{figure}
\begin{figure}[h!]
	\centering
		\includegraphics[width=0.8\textwidth]{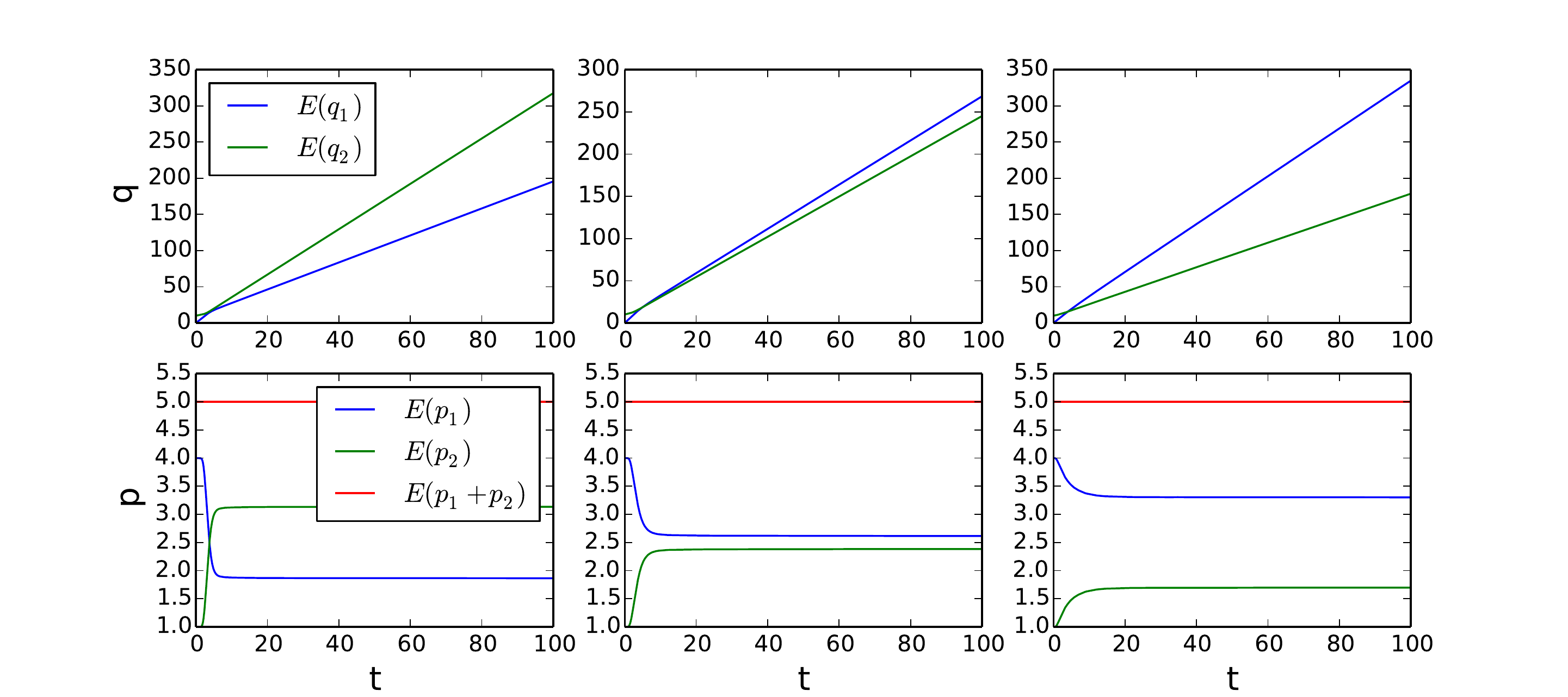}
		\caption{Numerical mean paths for Gaussian pulsons for the simulations with $\bar p_1=4$. Results for three example choices of the parameter $\beta$ are presented: $\beta=1.5$ ({\it left}), $\beta=2.5$ ({\it middle}), and $\beta=4.5$ ({\it right}). The positions are depicted in the plots in the upper row, and the corresponding momenta are shown in the plots in the lower row.}
		\label{fig: Mean paths 1}
\end{figure}
\begin{figure}[h!]
	\centering
		\includegraphics[width=0.7\textwidth]{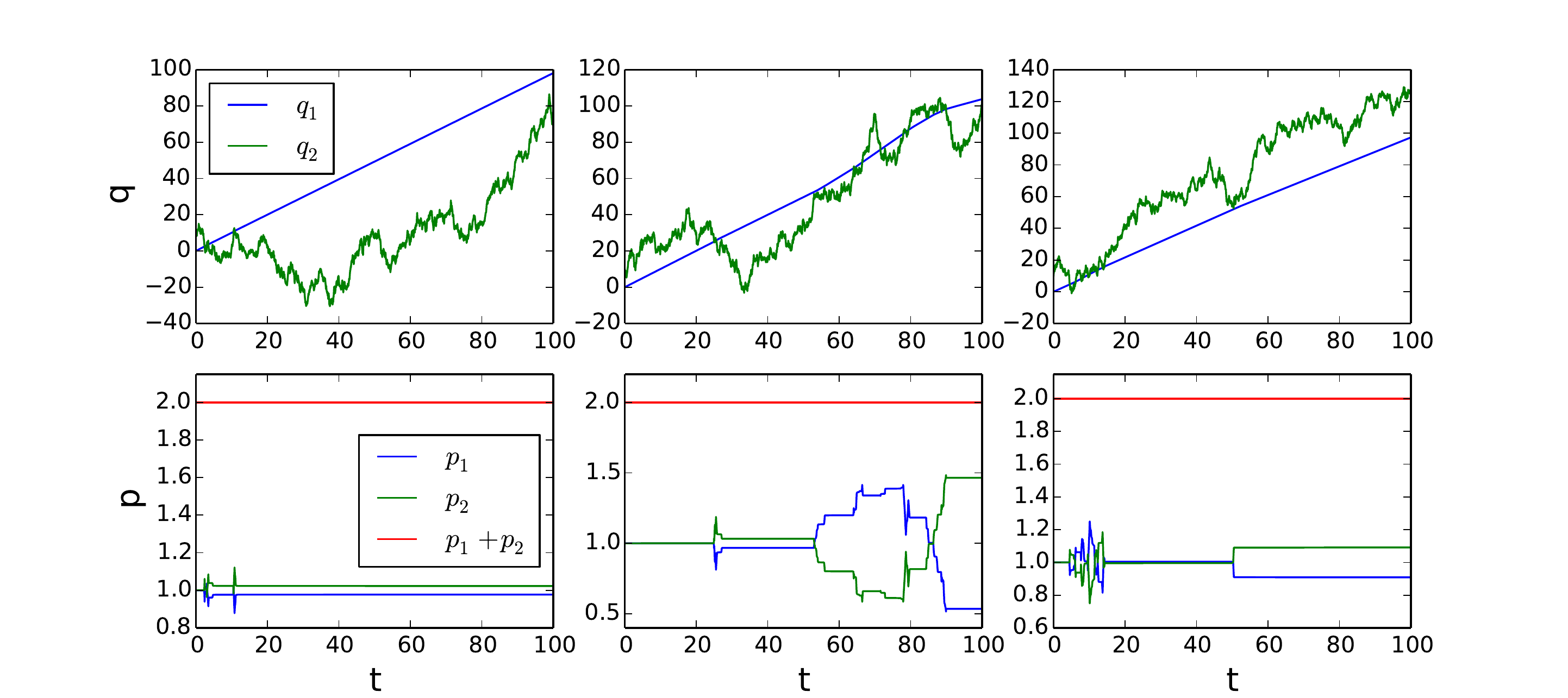}
		\caption{Example numerical sample paths for Gaussian pulsons for the simulations with $\bar p_1=1$ and $\beta=5$. The positions are depicted in the plots in the upper row, and the corresponding momenta are shown in the plots in the lower row.}
		\label{fig: Sample paths 2}
\end{figure}
\begin{figure}[h!]
	\centering
		\includegraphics[width=0.7\textwidth]{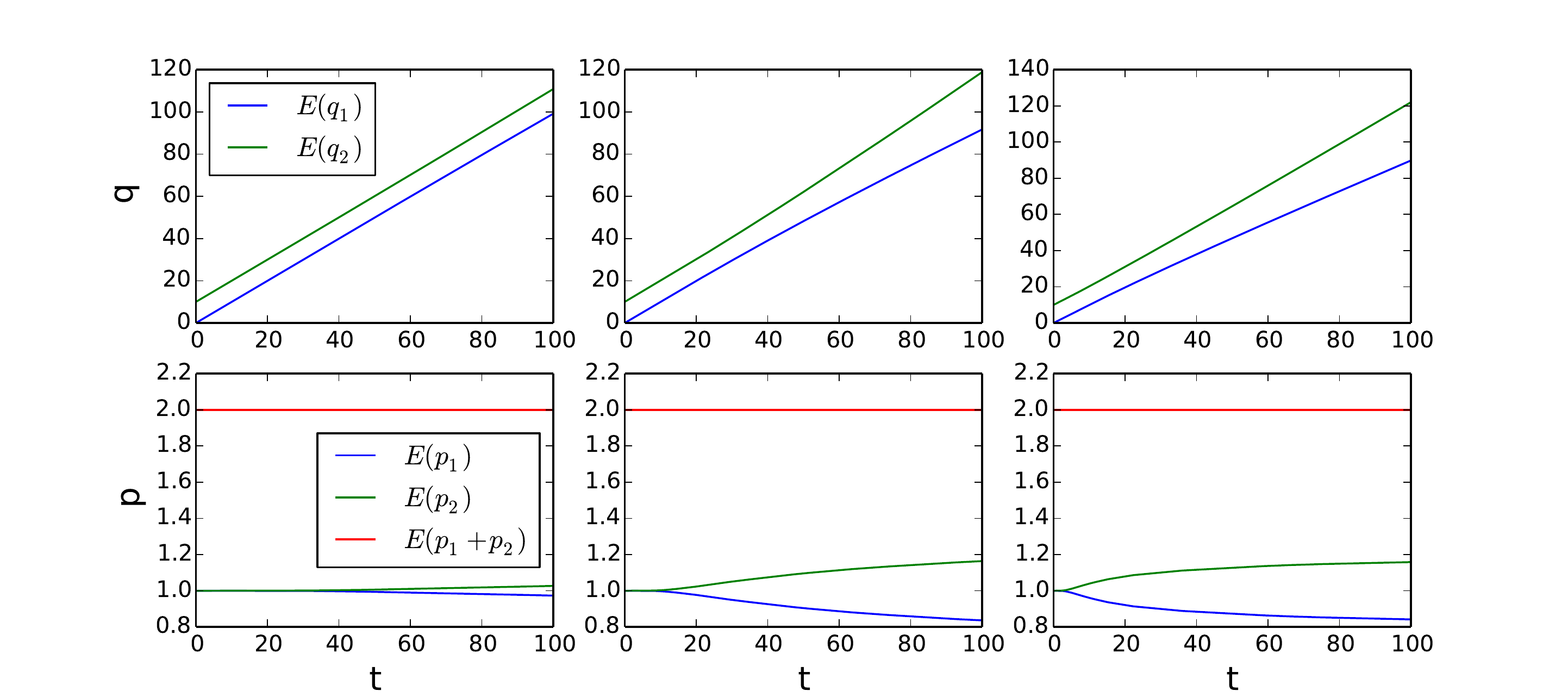}
		\caption{Numerical mean paths for Gaussian pulsons for the simulations with $\bar p_1=1$. Results for three example choices of the parameter $\beta$ are presented: $\beta=0.5$ ({\it left}), $\beta=1$ ({\it middle}), and $\beta=2$ ({\it right}). The positions are depicted in the plots in the upper row, and the corresponding momenta are shown in the plots in the lower row.}
		\label{fig: Mean paths 2}
\end{figure}
\begin{figure}
	\centering
		\includegraphics[width=0.7\textwidth]{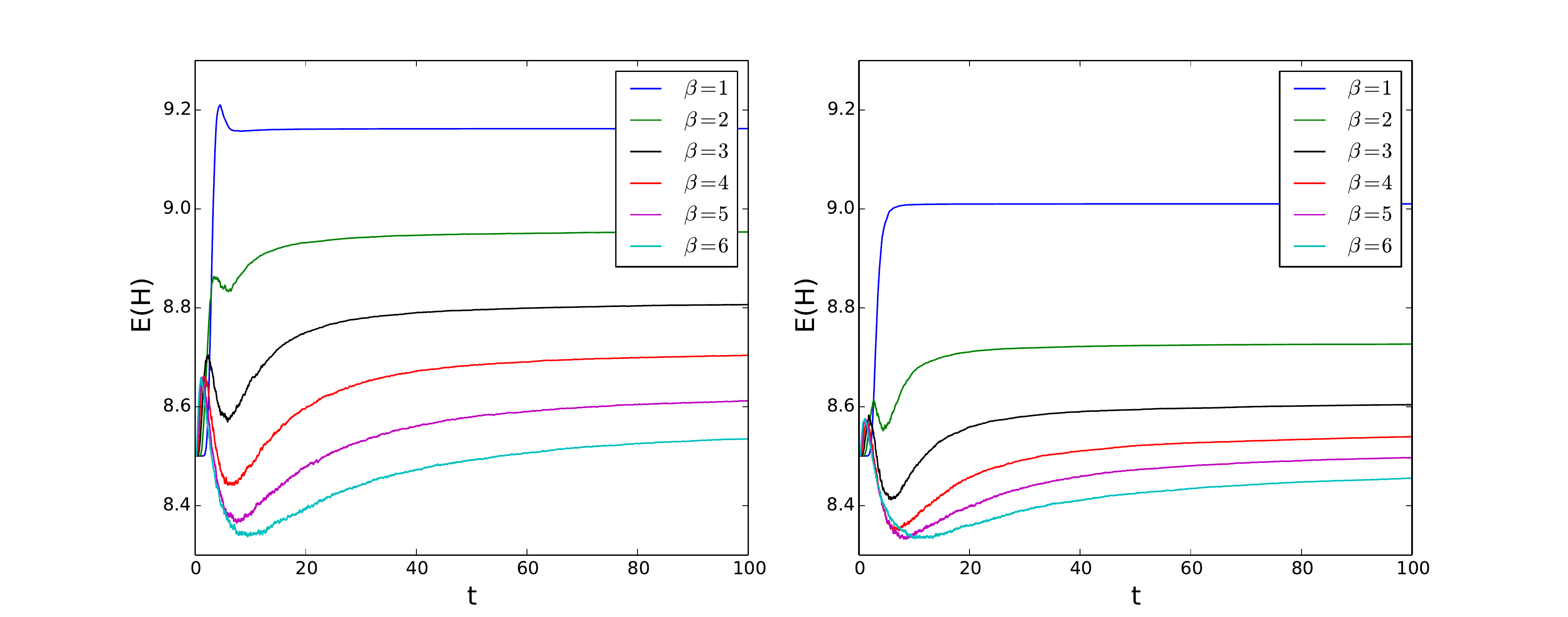}
		\caption{ Mean Hamiltonian for the simulations with $\bar p_1 = 4$ for Gaussian pulsons (\emph{left)} and peakons (\emph{right)}. }
		\label{fig: Hamiltonian}
\end{figure}
Figure~\ref{fig: Sample paths 1} shows a few sample paths from the simulations of the interaction of Gaussian pulsons for the case with $\bar p_1=4$ and $\beta=4$. The simulations for $\bar p_1=8$ and $\bar p_1=2$, as well as the simulations for peakons, gave qualitatively similar results. The most striking feature is that the faster pulson/peakon may in fact cross the slower one. In the deterministic case one can show that the faster pulson can never pass the slower one---they just exchange their momenta. The proof relies on the fact that both the Hamiltonian and total momentum are preserved (see \cite{FringerHolmGeodesic}, \cite{HolmGMS}). In our case, however, the Hamiltonian \eqref{qp-Ham} is not preserved due to the presence of the time-dependent noise (see Figure~\ref{fig: Hamiltonian}), which allows much richer dynamics of the interactions. This may find interesting applications in landmark matching---see the discussion in Section~\ref{sec: Summary and future work}.

Looking at Figure~\ref{fig: Sample paths 1} we also note that our variational integrator exactly preserves the total momentum, as expected. Figure~\ref{fig: Mean paths 1} depicts the mean solution for Gaussian pulsons with the initial condition $\bar p_1=4$ for different values of the noise intensity $\beta$. We see that for small noise the mean solution resembles the deterministic one, but as the parameter $\beta$ is increased, the mean solution represents two pulsons passing through each other with increasingly less interaction. We study the probability of crossing in more detail in Section~\ref{sec: Probability of crossing}.

We observed that pulsons may cross even when they have the same initial momentum (see Figure~\ref{fig: Sample paths 2}). In the deterministic case they would just propagate in the same direction, retaining their relative distance. Nevertheless, the mean solution (see Figure~\ref{fig: Mean paths 2}) does not show any crossing. 

\subsection{Probability of crossing}
\label{sec: Probability of crossing}

\begin{figure}
	\centering
		\includegraphics[width=0.8\textwidth]{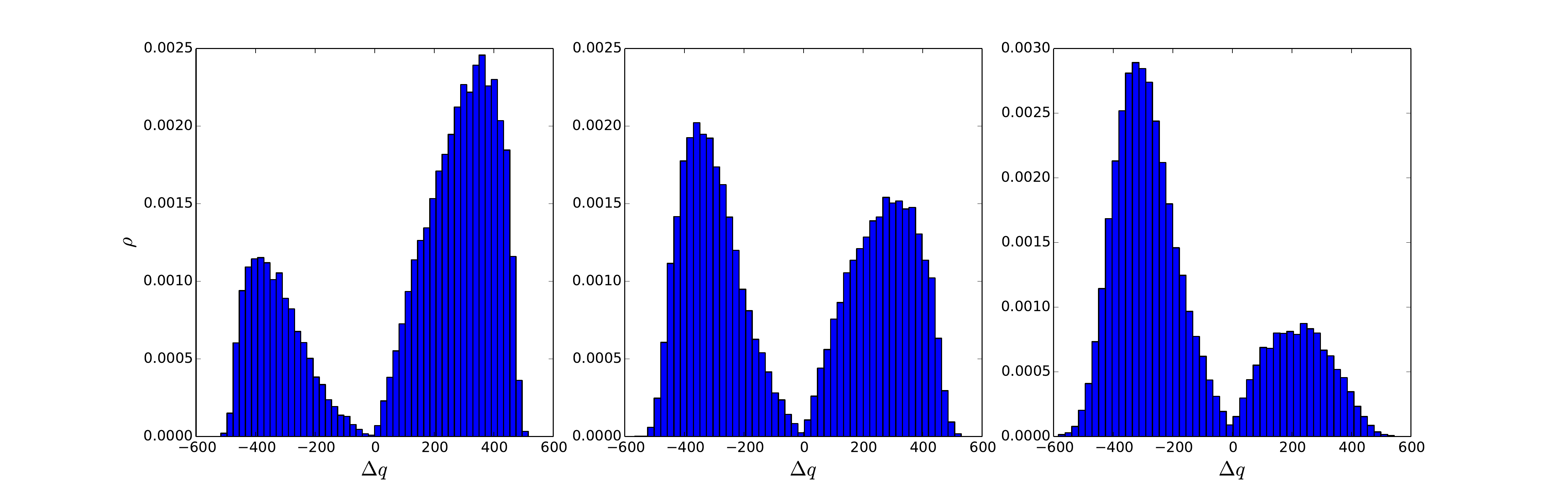}
		\caption{Numerical probability density $\rho$ of the distance $\Delta q(t)=q_2(t)-q_1(t)$ at time $t=100$ for Gaussian pulsons for the simulations with $\bar p_1=4$ (cf. Figure~\ref{fig: Mean paths 1}). Results for three example choices of the parameter $\beta$ are presented: $\beta=1.5$ ({\it left}), $\beta=2.5$ ({\it middle}), and $\beta=4.5$ ({\it right}). }
		\label{fig: Distribution Density 1}
\end{figure}

\begin{figure}
	\centering
		\includegraphics[width=0.8\textwidth]{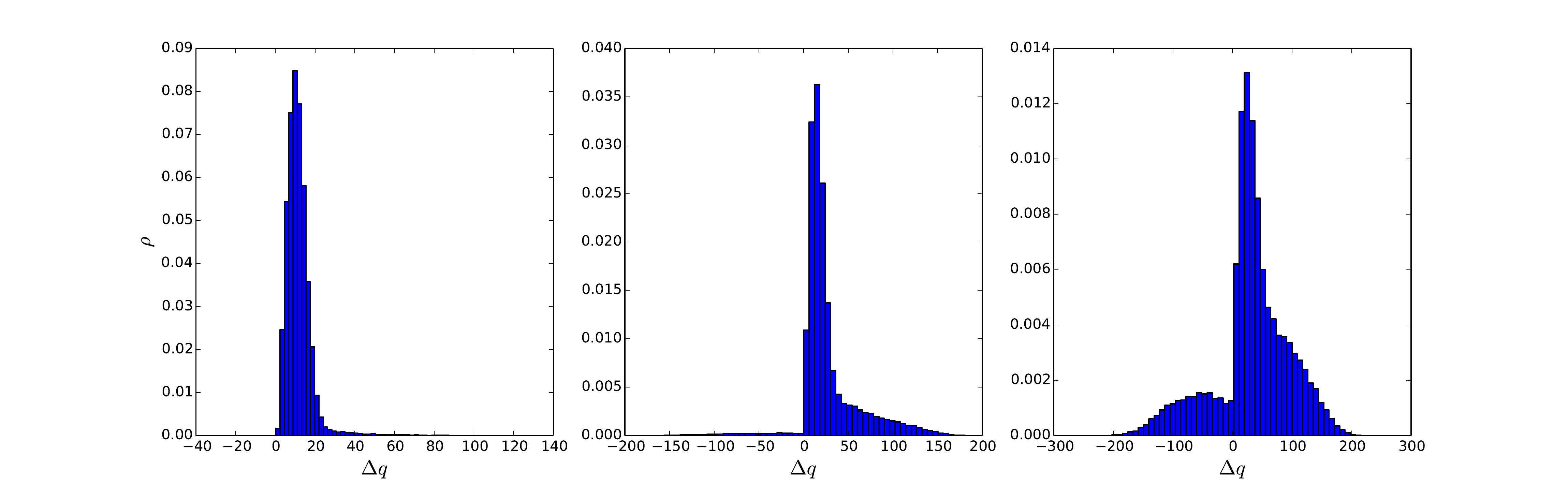}
		\caption{Numerical probability density $\rho$ of the distance $\Delta q(t)=q_2(t)-q_1(t)$ at time $t=100$ for Gaussian pulsons for the simulations with $\bar p_1=1$ (cf. Figure~\ref{fig: Mean paths 2}). Results for three example choices of the parameter $\beta$ are presented: $\beta=0.5$ ({\it left}), $\beta=1$ ({\it middle}), and $\beta=2$ ({\it right}). }
		\label{fig: Distribution Density 2}
\end{figure}

\begin{figure}
	\centering
		\includegraphics[width=0.8\textwidth]{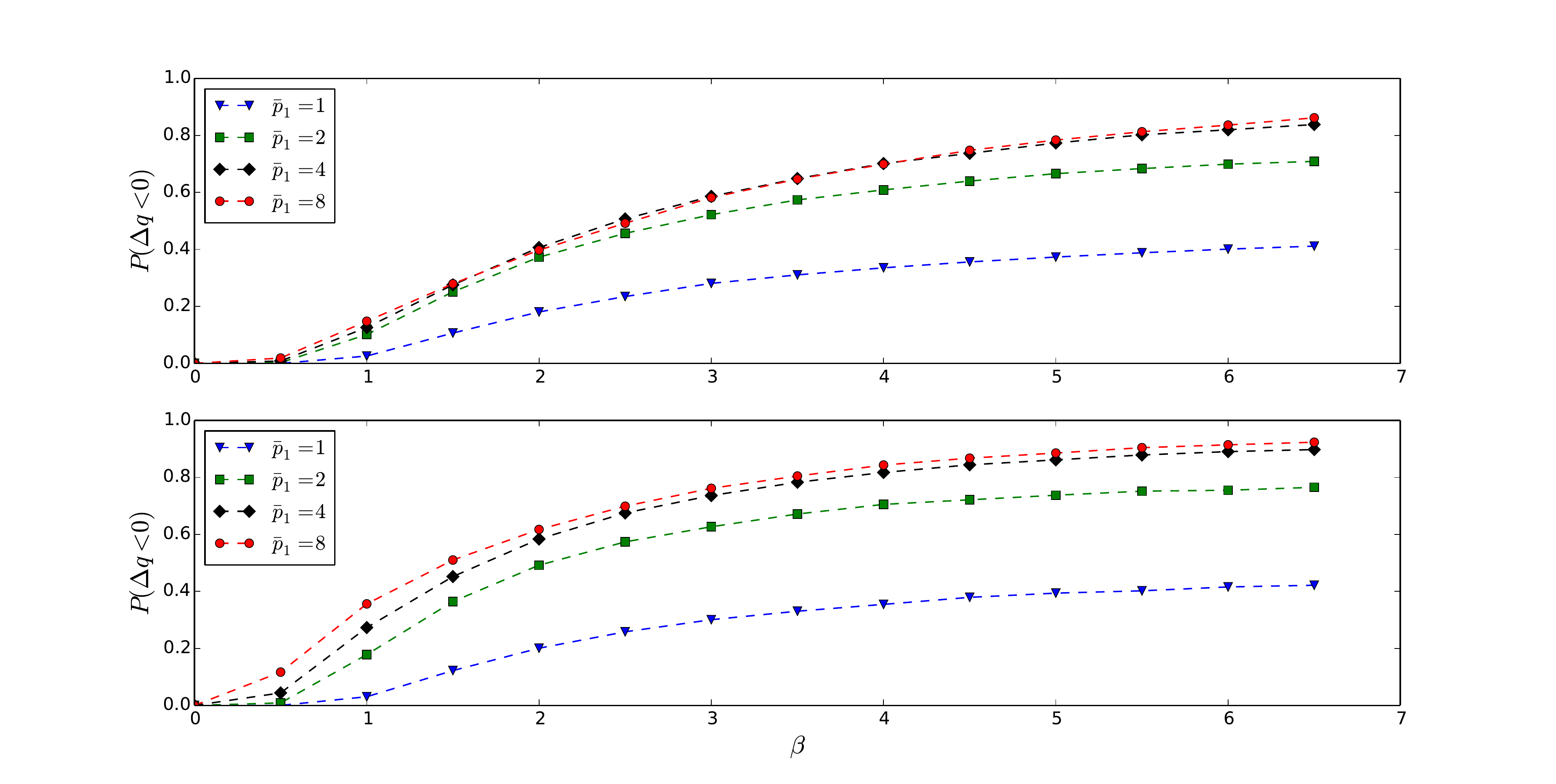}
		\caption{The probability of crossing, that is, the probability that $q_2(t)<q_1(t)$ at time $t=100$, as a function of the parameter $\beta$ for Gaussian pulsons ({\it top}) and peakons ({\it bottom}). }
		\label{fig: Crossing Probability}
\end{figure}

We studied in more detail the distance between the pulsons $\Delta q(t) = q_2(t)-q_1(t)$ at the end of the simulation, that is, at time $t=100$. Figure~\ref{fig: Distribution Density 1} presents the experimental probability density function of $\Delta q$ computed for Gaussian pulsons with $\bar p_1=4$. The density appears to have two local maxima, with the global maximum shifting from positive to negative values of $\Delta q$ as the noise intensity $\beta$ is increased. The simulations for $\bar p_1=8$ and $\bar p_1=2$, as well as the simulations for peakons, gave qualitatively similar results. Figure~\ref{fig: Distribution Density 2} depicts analogous results for the case of Gaussian pulsons with $\bar p_1=1$. In this case the density function also has two local maxima, but the global maximum never shifts to negative values of $\Delta q$. The probability of crossing as a function of the noise intensity $\beta$ is depicted in Figure~\ref{fig: Crossing Probability}. We see that this probability seems to approach unity for the simulations with $\bar p_1 >1$, and 0.5 for $\bar p_1 = 1$.

\subsection{First crossing time}
\label{sec: First crossing time}

\begin{figure}[h!]
	\centering
		\includegraphics[width=0.7\textwidth]{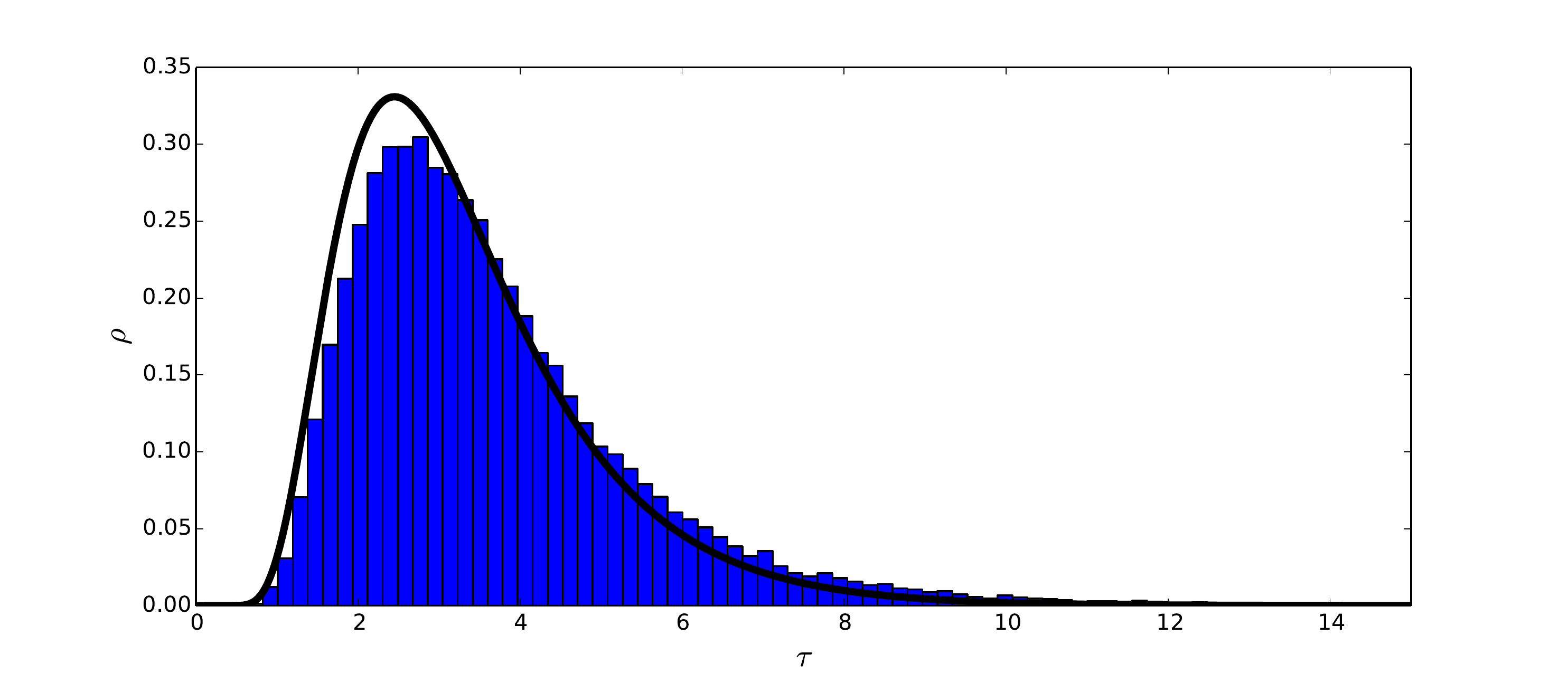}
		\caption{ The blue histogram presents the experimental probability density of the first crossing time $T_c$ for Gaussian pulsons for the simulations with $\bar p_1=4$ and $\beta=2.5$. More precisely, this is the conditional probability density given that $T_c<\infty$, i.e., assuming the pulsons do cross, the integral $\int_a^b \rho(\tau)\, d\tau$ yields the probability that the first crossing occurs at time $T_c \in [a,b]$. The black line depicts the inverse Gaussian distribution \eqref{eq: Inverse Gaussian distribution} with the parameters \eqref{eq: Mean and shape parameter}.  }
		\label{fig: Distribution density for the first crossing time}
\end{figure}
\begin{figure}[h!]
	\centering
		\includegraphics[width=0.9\textwidth]{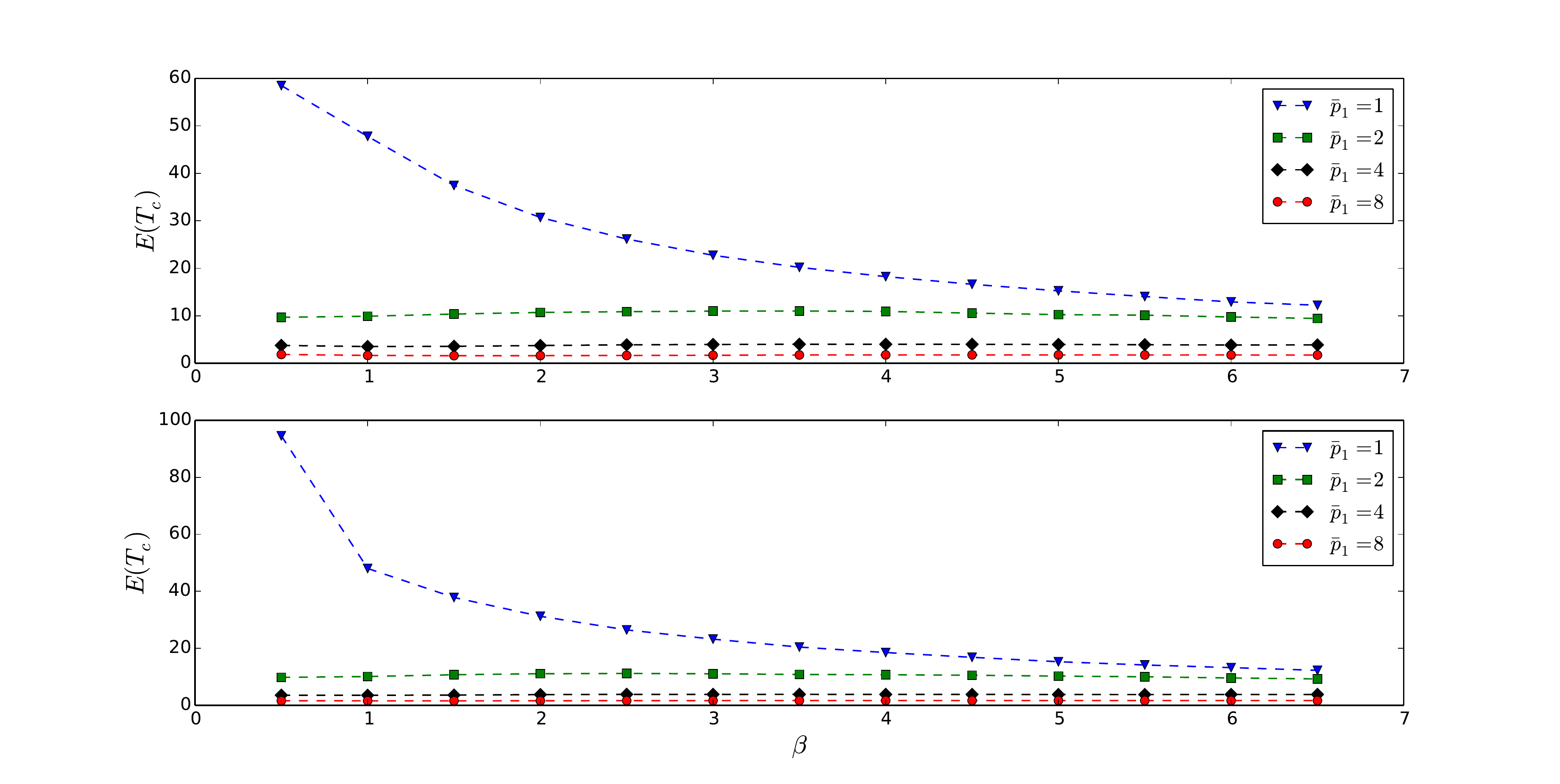}
		\caption{The mean first crossing time $E(T_c)$ as a function of the parameter $\beta$ for Gaussian pulsons ({\it top}) and peakons ({\it bottom}). More precisely, this is the conditional expectation $E(T_c | T_c<\infty)$ given that the pulsons do cross (i.e., $T_c<\infty$). }
		\label{fig: Mean first crossing time}
\end{figure}
It might be of interest to investigate the earliest time when the pulsons cross $T_c=\inf \{t>0: \Delta q(t)=0 \}$---let us call it the first crossing time (also known as the \emph{first exit time} or the \emph{hitting time}; see \cite{GardinerStochastic}, \cite{KloedenPlatenSDE}). Assume $T_c=+\infty$ if no crossing occurs. We can approximate the probability density of $T_c$ using the single-pulson solution \eqref{eq: Solution of the stochastic Hamiltonian equations for one pulson} and the asymptotic property \eqref{eq: Asymptotic solution to the F-P equation for two pulsons}. The two pulsons are initially far from each other, so we have

\begin{equation}
\label{eq: Asymptotic Delta q}
\Delta q(t) \approx \Delta \bar q - (\bar p_1 -\bar p_2) t + \beta W(t),
\end{equation}

\noindent
that is, $\Delta q(t)$ is approximated by a Brownian motion starting at $\Delta \bar q = \Delta q(0)>0$ with the drift $\bar p_1 -\bar p_2$. Assuming that $\bar p_1>\bar p_2$, the probability density of $T_c$ is given by the inverse Gaussian distribution $T_c \sim IG(\mu,\lambda)$ with the mean $\mu$ and shape parameter $\lambda$, respectively,

\begin{equation}
\label{eq: Mean and shape parameter}
\mu = \frac{\Delta \bar q}{\bar p_1 -\bar p_2}, \qquad\qquad \lambda=\frac{\Delta \bar q^2}{\beta^2},
\end{equation}

\noindent
where the density function is

\begin{equation}
\label{eq: Inverse Gaussian distribution}
\rho_{\text{IG}}(\tau) = \sqrt{\frac{\lambda}{2 \pi \tau^3}} \exp{\frac{-\lambda(\tau-\mu)^2}{2 \mu^2 \tau} }.
\end{equation}

\noindent
An example of a conditional probability density function of $T_c$ is depicted in Figure~\ref{fig: Distribution density for the first crossing time}. It shows very good agreement between the experimental and approximate theoretical distributions. The corresponding density functions for all other simulations look qualitatively similar, differing in the mean, variance, etc. When $\bar p_1 \longrightarrow \bar p_2^+$, then $\mu\longrightarrow +\infty$ and the inverse Gaussian distribution \eqref{eq: Inverse Gaussian distribution} tends to the L\'evy distribution, whose mean is infinite. However, we observed that for the simulations with $\bar p_1=\bar p_2=1$ the first crossing time is still well-approximated by the inverse Gaussian distribution with the shape parameter $\lambda$ as in \eqref{eq: Mean and shape parameter} and mean $\mu$ which seems to asymptotically decrease with the noise intensity $\beta$. The conditional mean crossing time $E(T_c | T_c<\infty)$ (given that a crossing occurs) as a function of the noise intensity $\beta$ is depicted in Figure~\ref{fig: Mean first crossing time}. The mean crossing time shows minor variations for the simulations with $\bar p_1 > 1$ and agrees well with \eqref{eq: Mean and shape parameter}, while it appears to asymptotically decrease for the simulations with $\bar p_1 = 1$.

\subsection{Noise screening}
\label{sec: Noise screening}

\begin{figure}[h!]
	\centering
		\includegraphics[width=0.7\textwidth]{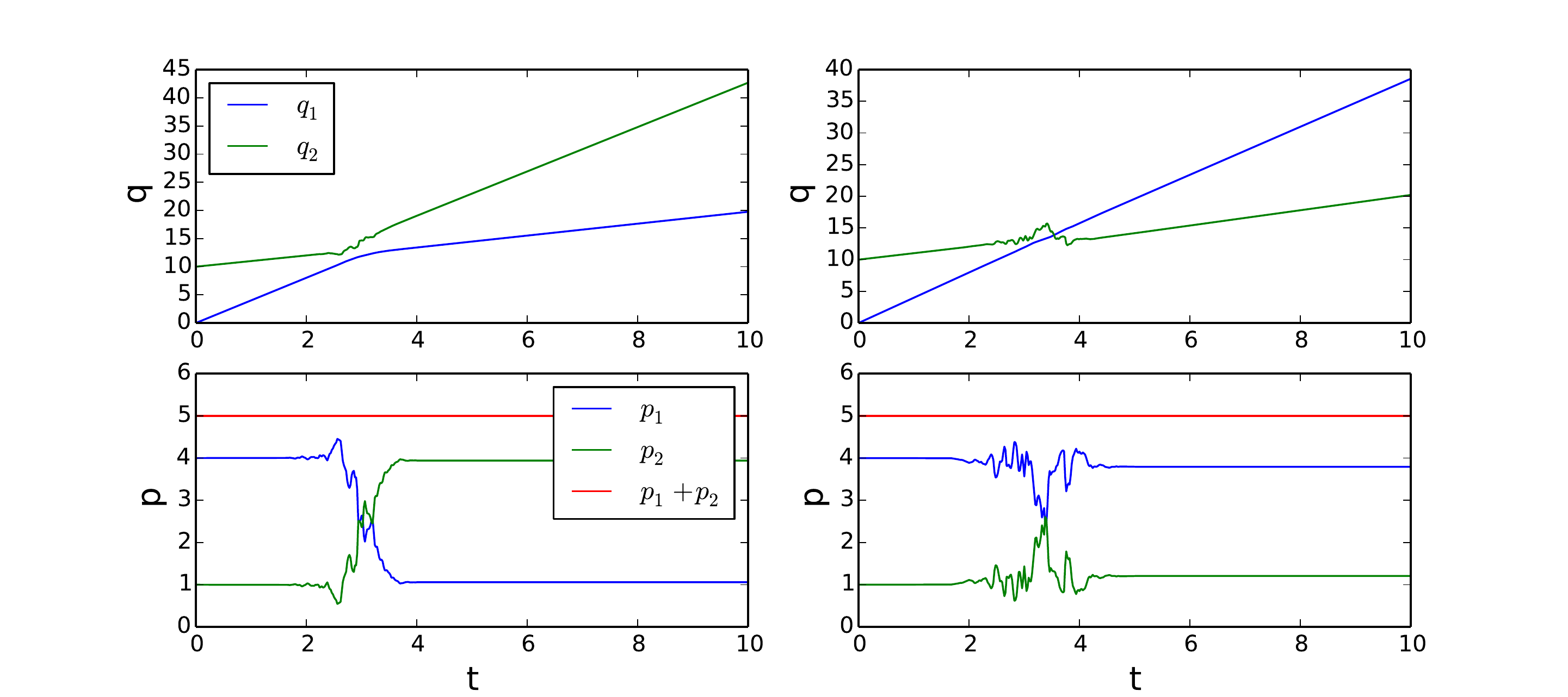}
		\caption{ Example numerical sample paths for Gaussian pulsons for the simulations with the initial conditions $\bar q_1 = 0$, $\bar p_1 = 4$, $\bar q_2 = 10$, and $\bar p_2 =1$, and the stochastic potential \eqref{eq: Noise screening stochastic potential} with the parameters $\beta=4$ and $\gamma=4$. The positions are depicted in the plots in the upper row, and the corresponding momenta are shown in the plots in the lower row. }
		\label{fig: Noise screening sample paths}
\end{figure}
In the numerical experiments described above we observed that the presence of noise causes pulsons to cross with a non-zero probability. The functions $q_1(t)$, $p_1(t)$, $q_2(t)$ and $p_2(t)$ define a transformation of the real line through \eqref{m-def-thm}. In the deterministic case this transformation is a diffeomorphism, but not when noise is added, since the crossing of pulsons introduces topological changes in the image of the real line under this transformation. This may be of interest in image matching, as in \cite{TrVi2012}, when one would like to construct a deformation between two images which are not exactly diffeomorphic. However, with that application in mind, one may want to restrict the stochastic effects only to the situation when two pulsons get close to each other. This can be obtained by applying the stochastic potential

\begin{equation}
\label{eq: Noise screening stochastic potential}
h(q,p) = \beta p_2 e^{-\frac{(q_2-q_1)^2}{\gamma}}.
\end{equation}

\noindent
The parameter $\beta \geq 0$ adjusts the noise intensity, just as before, while the parameter $\gamma>0$ controls the range over which the stochastic effects are non-negligible. We performed a few simulations with this stochastic potential. Since this potential is nonlinear, the integrator \eqref{eq: Stochastic Variational Integrator} is not applicable here. Instead, we used the stochastic symplectic midpoint rule (see \cite{MilsteinRepin}). A few sample paths are depicted in Figure~\ref{fig: Noise screening sample paths}. Note that this stochastic potential is translation-invariant, so the total momentum is preserved.

\subsection{Restriction to parametric noise and additive noise in the momentum equation}
\label{eq: Additive noise in the momentum equation}

\begin{figure}
	\centering
		\includegraphics[width=0.6\textwidth]{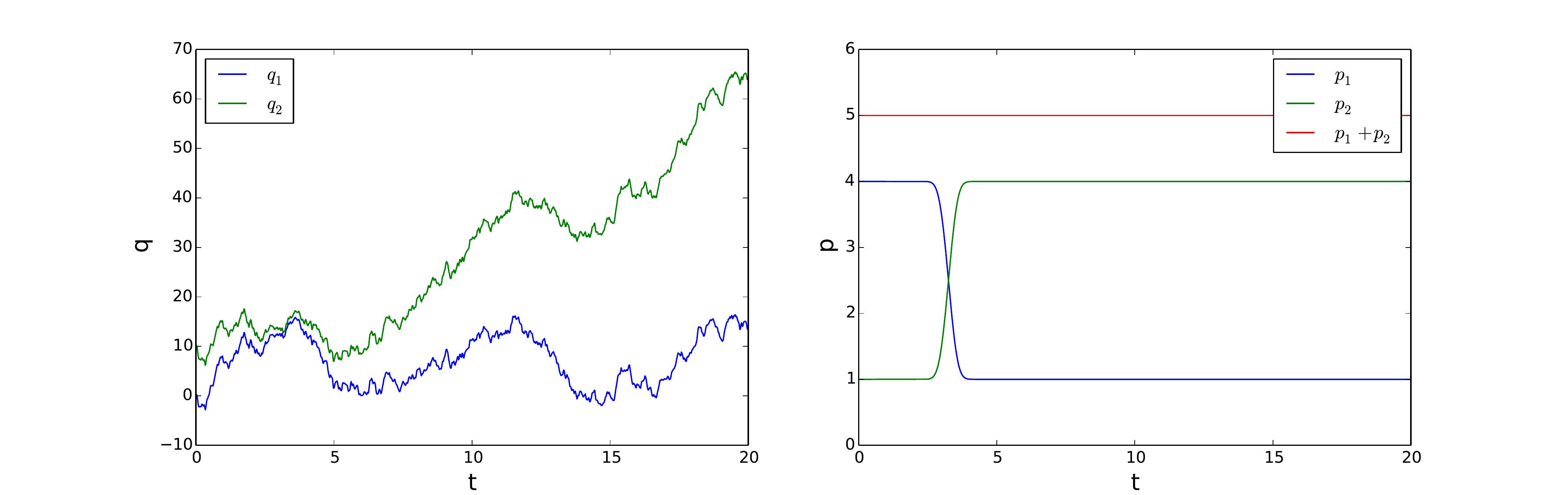}
		\caption{ Example numerical sample paths for Gaussian pulsons for the simulations with the initial conditions $\bar q_1 = 0$, $\bar p_1 = 4$, $\bar q_2 = 10$, and $\bar p_2 =1$, and the stochastic potential $h(q,p)=\beta (p_1+p_2)$ with the parameter $\beta=4$.}
		\label{fig: Parametric noise sample paths}
\end{figure}

\begin{figure}
	\centering
		\includegraphics[width=0.8\textwidth]{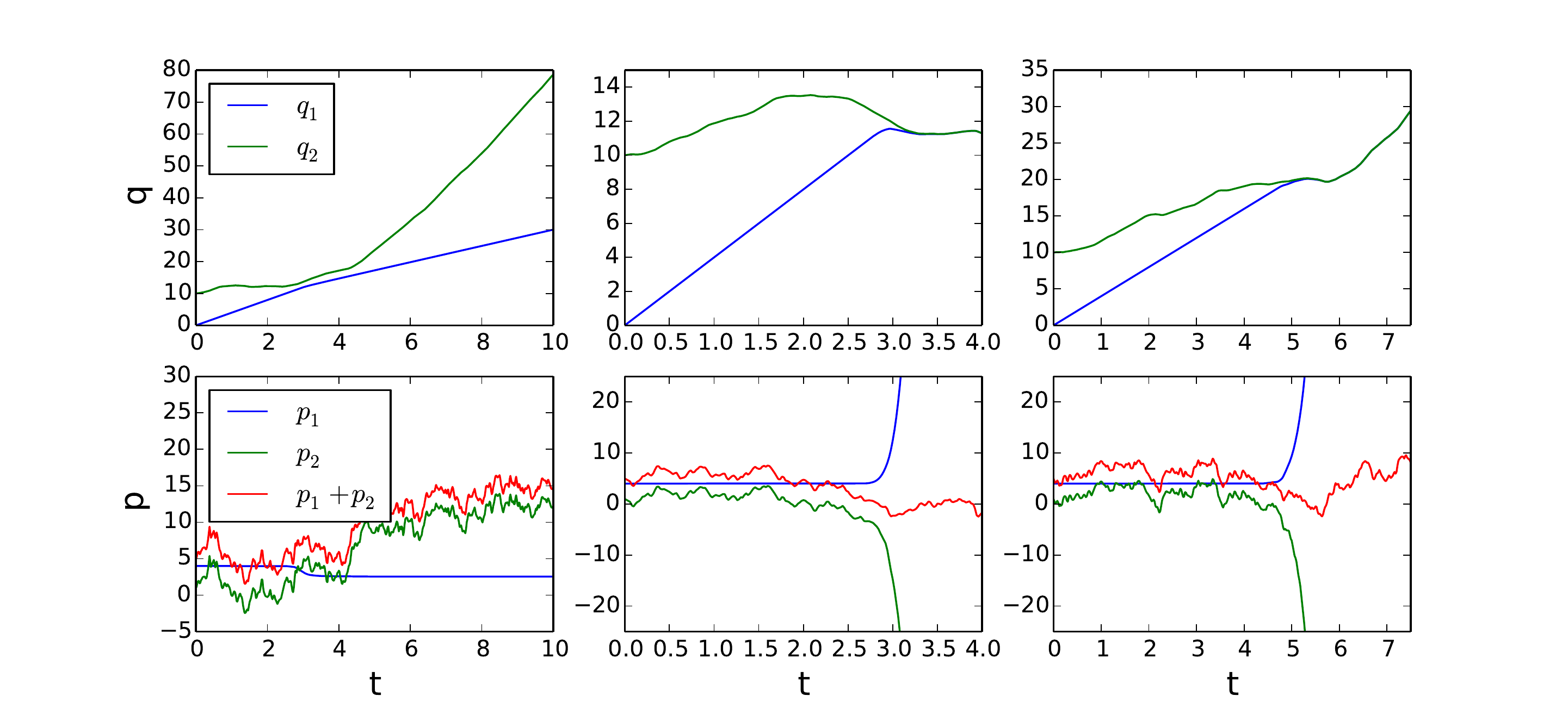}
		\caption{ Example numerical sample paths for Gaussian pulsons for the simulations with the initial conditions $\bar q_1 = 0$, $\bar p_1 = 4$, $\bar q_2 = 10$, and $\bar p_2 =1$, and the stochastic potential $h(q,p)=\beta q_2$ with the parameter $\beta=4$. The positions are depicted in the plots in the upper row, and the corresponding momenta are shown in the plots in the lower row. }
		\label{fig: Additive noise in q sample paths}
\end{figure}

Interestingly, crossing of pulsons does not seem to occur for the case of parametric stochastic deformation with the restriction $\varphi_{ia}(q) = \xi_i(q_a)$ as in Corollary~\ref{Lemma-m-eqn-peakon}. We ran numerical experiments for the potential $h(q,p)= \beta (p_1+p_2)$, which has the form as in Corollary~\ref{Lemma-m-eqn-peakon} with $\xi(x) = \beta$, but observed no interpenetration (see Figure~\ref{fig: Parametric noise sample paths}). This is consistent with our observation in Section~\ref{sec: Two-pulson dynamics with P-SD} and the fact that pulsons never cross in the deterministic case. We also did not observe crossing when the stochastic potential is independent of $p$. For instance, we performed simulations with the potential $h(q,p) = \beta q_2$. Such a potential results in additive noise in the momentum equation in \eqref{SEP-eqns-thm-qp} only, as in \cite{TrVi2012}. A few sample paths are depicted in Figure~\ref{fig: Additive noise in q sample paths}. Note that in this case the total momentum is not preserved, since $h(q,p)$ is not translationally invariant. In many cases the pulsons would asymptotically approach each other, but never pass. We observed similar behavior for the (translationally invariant) potential $h(q,p)=\beta \exp (-(q_1-q_2)^2 / \gamma)$ with $\beta, \gamma >0$.

\subsection{Convergence tests}
\label{sec: Convergence test}

\begin{figure}
	\centering
		\includegraphics[width=\textwidth]{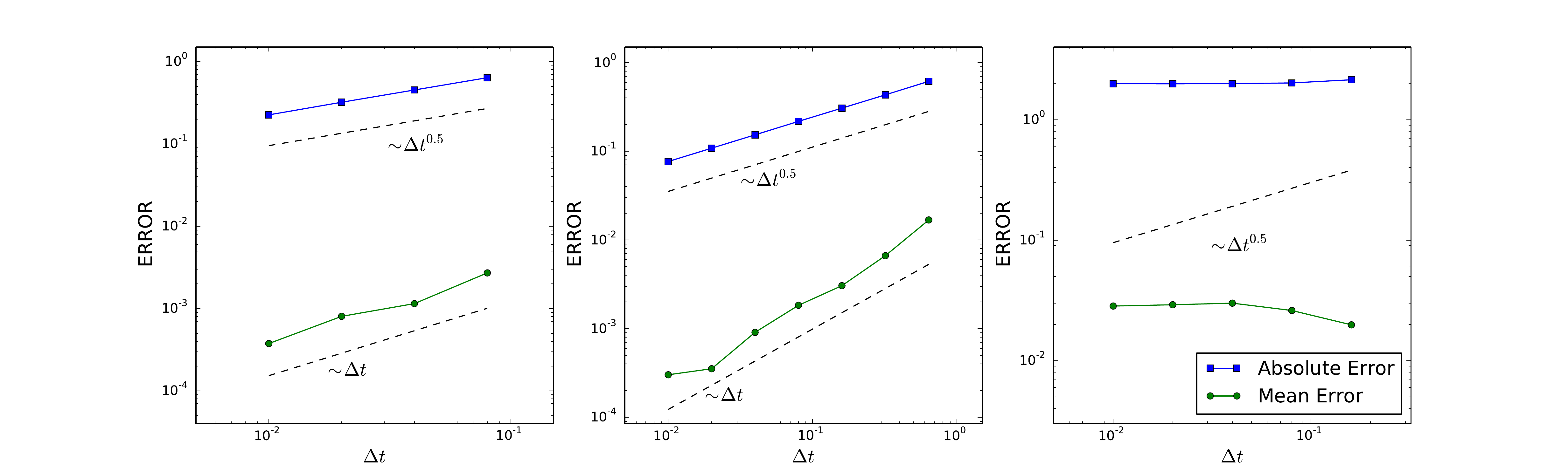}
		\caption{ Dependence of the absolute and mean errors on the time step for the single Gaussian pulson (\emph{left}), two point vortices (\emph{center}) and Kubo oscillator (\emph{right}). }
		\label{fig: Convergence plot}
\end{figure}

\subsubsection{Single pulson}
In order to test the convergence of the numerical algorithm \eqref{eq: Stochastic Variational Integrator} we performed computations for $N=1$ Gaussian pulson subject to one-dimensional (i.e., $M=1$) Wiener process with the stochastic potential $h(q,p)=\beta p$ (cf. Section~\ref{sec: Single-pulson dynamics}). Simulations with the initial conditions $\bar q=0$, $\bar p =4$ and the noise intensity $\beta=4$ were carried until the time $T=2$ for a number of decreasing time steps $\Delta t$. In each case 50,000 sample paths were generated. Let $z_{\Delta t}(t) = (q_{\Delta t}(t), p_{\Delta t}(t) )$ denote the numerical solution. We used the exact solution \eqref{eq: Solution of the stochastic Hamiltonian equations for one pulson} as a reference for computing the absolute error $E(|z_{\Delta t}(T)-z_\beta(T)|)$ and the mean error $| E(z_{\Delta t}(T))-E(z_\beta(T)) |$, where $z_\beta(t) = (q_\beta(t), p_\beta(t) )$. The dependence of these errors on the time step $\Delta t$ is depicted in Figure~\ref{fig: Convergence plot}. We verified that our algorithm has strong order of convergence $0.5$, and weak order of convergence 1.

\subsubsection{Two planar point vortices}

We performed a similar test for $N=2$ planar point vortices subject to a one-dimensional (i.e., $M=1$) Wiener process. The system is described by

\begin{equation}
H(q,p)=-\frac{1}{4 \pi} \Gamma_1 \Gamma_2 \log \bigg[ \Big( \frac{q_1}{\sigma_1}-\frac{q_2}{\sigma_2} \Big)^2 + \Big( \frac{p_1}{\lambda_1}-\frac{p_2}{\lambda_2} \Big)^2 \bigg], \qquad\quad h(q,p) = \beta \lambda_1 q_1 + \gamma \sigma_1 p_1 + \beta \lambda_2 q_2 + \gamma \sigma_2 p_2, 
\end{equation}

\noindent
where $\Gamma_1, \Gamma_2$ are the circulations of the vortices, $\sigma_i = \sqrt{|\Gamma_i|}\,\text{sgn}\,\Gamma_i$, $\lambda_i = \sqrt{|\Gamma_i|}$ are scaling factors, $\beta, \gamma$ are the noise intensities, and $q_i, p_i$ denote the $x$- and $y$-coordinate of the $i$-th vortex, respectively (see \cite{Flandoli}, \cite{Newton}). Simulations for $\Gamma_1=2$, $\Gamma_2=1$ with the initial conditions $\bar q_1 = \sigma_1 R_1$ , $\bar q_2 = \sigma_2 R_2$, $\bar p_1=\bar p_2=0$, where $R_1=\Gamma_2/(\Gamma_1+\Gamma_2)$, $R_2=-\Gamma_1/(\Gamma_1+\Gamma_2)$, and the noise intensities $\beta=\gamma=0.5$ were carried out until the time $T=6.4$ for a number of decreasing time steps $\Delta t$. In each case 50000 sample paths were generated. We used the exact solution (see \cite{Flandoli}, \cite{Newton})

\begin{align}
q_1(t) &= \sigma_1 R_1 \cos \omega t + \gamma \sigma_1 W(t), \qquad p_1(t) = \lambda_1 R_1 \sin \omega t - \beta \lambda_1 W(t), \\
q_2(t) &= \sigma_2 R_2 \cos \omega t + \gamma \sigma_2 W(t), \qquad p_2(t) = \lambda_2 R_2 \sin \omega t - \beta \lambda_2 W(t),
\end{align}

\noindent
where $\omega = (\Gamma_1+\Gamma_2)/(2 \pi)$, as a reference for computing the absolute and mean errors (see Figure~\ref{fig: Convergence plot}). We verified that our algorithm has strong order of convergence $0.5$, and weak order of convergence 1.

\subsubsection{Kubo oscillator}

To demonstrate that the integrator \eqref{eq: Stochastic Variational Integrator} fails to converge for multiplicative noise, we performed computations for the Kubo oscillator, which is defined by $H(q,p)=p^2/2+q^2/2$ and $h(q,p)=\beta(p^2/2+q^2/2)$, where $\beta$ is the noise intensity (see \cite{MilsteinRepin}). The exact solution is given by 

\begin{equation}
q(t)=\bar p \sin(t+\beta W(t)) + \bar q \cos(t+\beta W(t)), \qquad\quad p(t)=\bar p \cos(t+\beta W(t)) - \bar q \sin(t+\beta W(t)).
\end{equation}

\noindent
A similar convergence test with $\bar q=0$, $\bar p=4$, $\beta=1$, and $T=6.4$ revealed that the integrator \eqref{eq: Stochastic Variational Integrator} failed to converge, although the errors remained bounded (see Figure~\ref{fig: Convergence plot}).

\section{Summary}
\label{sec: Summary and future work}
We have seen in Section \ref{sec: ParametricStochasticDeformations} that the finite-dimensional peakon solutions for the EPDiff partial differential equation in one spatial dimension persist under both parametric stochastic deformation (P-SD) and canonical Hamiltonian stochastic deformations (CH-SD) of the EPDiff variational principle. Being both finite-dimensional and canonically Hamiltonian, the dynamics of the peakon solution set for EPDiff admits the entire range of CH-SD in the sense of \cite{Bi1981,LaCa-Or2008}, which includes P-SD but can be more general. Therefore, the peakon solution set offers a finite-dimensional laboratory for comparing the effects of P-SD and CH-SD on the stochastically deformed EPDiff SPDE solution behaviour. In fact, as it turns out, the peakon solution set for EPDiff offers a particularly sensitive assessment of the effects of stochasticity on finite-dimensional solutions of SPDE. In Section \ref{sec: ParametricStochasticDeformations}, we took advantage of the flexibility of CH-SD to study stochastic peakon-peakon collisions in which noise was introduced into \emph{only one} of the peakon position equations (rather than symmetrically into both of the canonical position equations, as occurs with P-SD), while at the same time not introducing any noise into either of the corresponding canonical momentum equations. The precision and flexibility of the CH-SD approach to stochastic peakon-peakon collision dynamics revealed that its asymmetric case with noise in only one canonical position equation allows the soliton-like singular peakon and pulson solutions of EPDiff to interpenetrate and change order on the real line, although this is not possible for the diffeomorphic flow represented by the solutions of the unperturbed deterministic EPDiff equation. This crossing of peakon paths was observed and its statistics were studied in detail for CH-SD in numerical experiments in Section~\ref{sec: Numerical experiments}. In contrast, crossing of peakon paths was \emph{not} observed for the corresponding P-SD simulations in which the noise enters symmetrically in both position equations. Crossing of peakon paths was also not observed when stochasticity was added only in the canonical momentum equations, as studied in \cite{TrVi2012}. 

Thus, for the deterministic EPDiff, adding stochasticity of \emph{constant amplitude} with either CH-SD of P-SD to a finite dimensional invariant solution set has been found to produce different SDE solution behaviour. Here, the difference has introduced the possibility of a topological change in the order of points moving on a line in the CH-SD approach, while no such change in topology seems to be available via the P-SD subclass. One can also imagine that changing the level of noise in the P-SD EPDiff SPDE could change the number of peaks or pulsons; a feature which would not have been available if the level of noise were changed after the reduction to a fixed $N$-peakon solution sector. The intriguing idea of creation of singular EPDiff solutions by P-SD noise in the SPDE is under current investigation. 

The investigation of stochastic EPDiff in this paper has raised and illustrated a potentially important issue. The need for assessing the validity of approximating the stochastic solution behaviour of nonlinear SPDE by SDE obtained from adding noise to finite-dimensional projections (or discretisations) of the solutions is likely to be encountered quite often in many other circumstances and can be expected to be of frequent future concern. In particular, this issue is likely to occur in considerations of model error in stochastic data assimilation. For example, the endeavours of computational anatomy must face this issue in the use of the singular solutions of EPDiff known as \emph{landmarks} in the task of registration of noisy images \cite{TrVi2012}. This issue of the validity of stochastic deformations of finite-dimensional approximations of evolutionary PDE is a challenge for continuing research in P-SD of EPDiff, as well as in stochastic deformations of more general continuum equations, such as Euler's equations for an ideal fluid, or the Navier-Stokes equations for a viscous fluid. The present work has shown that the introduction of even constant stochasticity into the equations of motion for \emph{exact solutions} (peakons and pulsons, or landmarks, for EPDiff) can produce unexpected changes in topology of the solution in one dimension. The corresponding introduction of stochasticity into the equations of motion for finite-dimensional \emph{approximations} such as discretisation, or projections of the solutions of nonlinear evolutionary PDE may result in  other surprises. 

\subsection{Acknowledgements} 
We are very grateful for the encouragement of the many people who took the time to discuss these matters with us, or comment on drafts, especially our friends and colleagues N. Bou-Rabee, A. Castro, M. Chekroun, C. J. Cotter, D. Crisan, M. O. Hongler, J. P. Ortega and H. Owhadi. However, as usual, any mistakes belong to the authors. This work was partially supported by the European Research Council Advanced Grant 267382 FCCA.

%


\end{document}